\documentclass{llncs}
\pdfoutput=1
\usepackage{enumerate}
\usepackage{amssymb,amsmath,latexsym}
\usepackage{mathtools}
\usepackage{enumitem}
\usepackage{array}\usepackage{graphicx}
\usepackage[all,line]{xy}
\usepackage{color}

\newcommand{\power}[1]{\ensuremath{\mathcal{P}(#1)}}
\newcommand{\st}[1]{\ensuremath{\mathbf{states}(#1)}}
\newcommand{\union}{\, \cup \,}
\newcommand{\lang}{\ensuremath{\mathcal{L}}}
\newcommand{\equivalence}{\leftrightarrow}
\newcommand{\cl}[1]{\ensuremath{\mathsf{cl}(#1)}}
\newcommand{\hintikka}[1]{\ensuremath{\mathcal{#1}}}
\newcommand{\inter}{\, \cap \,}
\newcommand{\tableau}[1]{\ensuremath{\mathcal{#1}}}
\newcommand{\fm}[1]{\emph{#1}}
\newcommand{\subf}[1]{\ensuremath{\mathsf{Sub}(#1)}}
\newcommand{\ecl}[1]{\ensuremath{\mathsf{ecl}(#1)}}
\newcommand{\bigo}[1]{\ensuremath{\mathcal{O}(#1)}}
\newcommand{\card}[1]{\ensuremath{|#1|}}
\newcommand{\RRule}[1]{\textbf{(#1)}}
\newcommand{\powerne}[1]{\ensuremath{\mathcal{P}^{\tiny +}(#1)}}
\newcommand{\imp}{\rightarrow}
\newcommand{\knows}[1]{\ensuremath{\mathbf{K}}_{#1}}
\newcommand{\crh}[2]{\ensuremath{\{\, #1 \mid \, #2\, \}}}
\newcommand{\kframe}[1]{\ensuremath{\mathfrak{#1}}}
\newcommand{\set}[1]{\ensuremath{ \{ #1 \} }}
\newcommand{\rel}[1]{\ensuremath{\mathcal{#1}}}
\newcommand{\bigunion}{\bigcup \,}
\newcommand{\biginter}{\bigcap \,}
\newcommand{\fexp}{\textsc{FullExpansion}}
\newcommand{\system}[1]{\ensuremath{\mathbf{#1}}}
\newcommand{\MAELC}{\textbf{MAEL(C)}}
\newcommand{\MAELD}{\textbf{MAEL(D)}}
\newcommand{\MAELCD}{\textbf{MAEL(CD)}}
\newcommand{\ATEL}{\textbf{ATEL}}
\newcommand{\vp}{\ensuremath{\varphi}}
\newcommand{\ATL}{\textbf{ATL}}
\newcommand{\PDL}{\textbf{PDL}}
\newcommand{\CMAELC}{\textbf{CMAEL(C)}}
\newcommand{\CTL}{\textbf{CTL}}
\newcommand{\LTL}{\ensuremath{\textbf{LTL}}}
\newcommand{\ap}{\textbf{\texttt{AP}}}
\newcommand{\con}{\wedge}

\newcommand{\de}[1]{\emph{#1}}
\newcommand{\commonk}[1]{\ensuremath{\mathbf{C}}_{#1}}
\newcommand{\distrib}[1]{\ensuremath{\mathbf{D}}_{#1}}

\newcommand{\cmaelcd}{\textbf{CMAEL}(CD)}

\newcommand{\DtoD}[4]{#1\stackrel{\neg \distrib{#3} #4}{\longrightarrow} #2}
\newcommand{\FE}{\text{CS}}
\newcommand{\funny}{cut-saturated }
\newcommand{\di}[1]{\distrib{#1}}
\newcommand{\co}[1]{\commonk{#1}}
\newcommand{\De}{\ensuremath{\Delta}}
\newcommand{\Ga}{\ensuremath{\Gamma}}

\newcommand{\fe}[1]{\ensuremath{\mathcal{FE}(#1)}}
\newcommand{\cse}[1]{\ensuremath{\mathcal{CSE}(#1)}}

\newcommand{\CMAELCD}{\textbf{CMAEL(CD)}}
\newcommand{\mmodel}[1]{\ensuremath{\mathcal{#1}}}
\newcommand{\sat}[3]{\ensuremath{\mmodel{#1}, #2 \Vdash #3}}
\newcommand{\agents}{\ensuremath{\Sigma}}
\newcommand{\MAEL}{\textbf{MAEL}}

\newcommand{\ttt}[1]{\texttt{#1}}

\title{Tableau-based
  decision procedure for the multiagent epistemic logic with all 
  coalitional operators for common and distributed knowledge}
\titlerunning{Tableaux for the full coalitional multiagent epistemic logic} 
 \author{Mai Ajspur\inst{1} \and Valentin Goranko\inst{2} \and  Dmitry Shkatov\inst{3}}
\authorrunning{Mai Ajspur, Valentin Goranko and  Dmitry Shkatov} 
\institute{Roskilde University, \email{ajspur@ruc.dk}, \\
\and Technical University of Denmark and University of Johannesburg\thanks{
{Visiting professor}}, \email{vfgo@imm.dtu.dk} \\
\and 
University of the Witwatersrand, \email{dmitry@cs.wits.ac.za}}
\begin{document}

\maketitle  
\pagestyle{plain}

\centerline{January 16, 2012}
%\begin{paper}
\begin{abstract}
  We develop a conceptually clear, intuitive, and feasible decision procedure for testing satisfiability in the full multi\-agent epistemic logic \CMAELCD\ with operators for common and distributed knowledge for all coalitions of agents mentioned in the language. To that end, we introduce Hintikka
 structures for \CMAELCD\ and prove that satisfiability in such structures is equivalent to satisfiability in standard models. Using that result, we design an incremental tableau-building procedure that eventually constructs a satisfying Hintikka structure for every satisfiable input set of formulae of \CMAELCD\ and closes for every unsatisfiable input set of formulae. 
\end{abstract}
\keywords{multi-agent epistemic logic, satisfiability, tableau, decision procedure}

%%%%%%%%%%%%%%%%%%%%%%%%%%%%%%%%%%%%%%%%%%%%%%
\section{Introduction}
\label{sec:intro}

Over the last three decades, multiagent epistemic
logics~\cite{Fagin95knowledge}, \cite{vdHM95} have been playing an
increasingly important role in computer science and AI.  The
earliest prominent applications have been to specification, design, and
verification of distributed protocols ~\cite{Halpern87} and
\cite{HM90}; a number of other applications are described in, among
others,~\cite{Fagin95knowledge}, \cite{FHV92}, and \cite{vdHM95}.  The
most recent, and perhaps more important ones are to specification,
design, and verification of \emph{multiagent systems} --- a research area that
has emerged on the borderline between distributed computing, AI, and
game theory \cite{Shoham08}, \cite{Weiss99}, \cite{Wooldridge02}.

\subsection{Multiagent epistemic logics and decision methods for them} 
Languages of multiagent epistemic logics considered in the literature
contain various repertoires of epistemic operators. We refer to the
basic multiagent epistemic logic, containing only operators of
individual knowledge for a finite non-empty set $\agents$ of agents,
as \MAEL\ (\textbf{M}ulti \textbf{A}gent \textbf{E}pistemic
\textbf{L}ogic).  Since all epistemic operators of this logic are
\system{S5}-type modalities, it is also referred to in the literature
as $\system{S5}_n$, where $n$ is the number of agents in the language.
The logic obtained from \MAEL\ by adding the operator of common
knowledge among all agents in $\agents$ is then called \MAELC. This
logic, along with \MAEL, was studied in~\cite{HM92}.  Analogously, if
\MAEL\ is augmented with the operator of distributed knowledge for all
agents, then the resulting logic will be called \MAELD. It was studied
in~\cite{FHV92} and \cite{vdHM92}.
\MAEL\ augmented with operators of both common and distributed
knowledge for the set of all agents, hereafter called \MAELCD, was
studied in~\cite{vdHM97}, and a tableau-based decision procedure for
it was first presented in ~\cite{GorSh08a}.  Thus, all logics mentioned so
far either do not have both operators of common and distributed
knowledge, or only have those operators for the whole set of agents in
the language.

At the same time, there has recently been an increasing interest in
the study of \emph{coalitional} multiagent logics (see~\cite{Pauly01},
\cite{Pauly01a}, \cite{Pauly02}, \cite{AHK02}, \cite{vdHW04},
\cite{GorJam04}), i.e. logics whose languages refer to any groups
(\emph{coalitions}) of agents. These are important, inter alia, in
multiagent systems, where agents may ``cooperate'' (i.e., form a
coalition) in order to achieve a certain goal.
Most of the so far studied logical formalisms referring to coalitions
of agents have only been concerned with formalizing reasoning about
\emph{strategic abilities} of coalitions.  (A notable exception
is~\cite{vdHW04}, where the Alternating-time Temporal Epistemic Logic
\ATEL\ was introduced, whose language contains both common knowledge
and strategic operators for coalitions of agents.) Clearly, real
cooperation can only be achieved by communication, i.e., exchange of
knowledge. Thus, it is particularly natural and important to consider
multiagent epistemic logics with operators for both common and
distributed knowledge among \emph{any (non-empty) coalitions of
  agents}. This is the logic under consideration in the present paper, hereby
called \CMAELCD\ (for \textbf{C}oalitional
\textbf{M}ulti-\textbf{A}gent \textbf{E}pistemic \textbf{L}ogic with
operators of \textbf{C}ommon and \textbf{D}istributed knowledge). It subsumes
all multiagent epistemic logics mentioned above, except \ATEL. 

In order to be practically useful for such tasks as specification and design of distributed or multiagent systems, the respective logic need to be equipped with algorithms solving (constructively) its satisfiability problem, i.e.  testing whether a given input formula $\vp$ of that logic is satisfiable and, if so, providing enough information for the construction of a model for $\vp$.
Decidability of modal logics, including epistemic logics, is usually 
proved by establishing a `small model property', which provides a brute force decision procedure consisting of exhaustive search for a model amongst all those whose size is within the theoretically prescribed bounds. The two most common practically feasible general methods for satisfiability checking of modal logics are based on \emph{automata}
\cite{Vardi06} and on tableaux (see e.g., \cite{Pratt79}, \cite{BPM83}, \cite{Fitting83}, \cite{Wolper85}, \cite{EmHal85}, \cite{Gore98},
\cite{Fitting07}). 

There are various styles of tableau-based decision procedures; see \cite{Fitting83}, \cite{Gore98} and \cite{Fitting07} for detailed exposition and surveys. An easy to describe but somewhat less efficient and practically unfeasible approach, that we will call \emph{maximal tableau} (also  called \emph{top-down} in \cite{EmHal85}), consists in trying to build in one step a `canonical' finite model for any given formula out of \emph{all}  maximal consistent subsets of the closure of that formula. This method always works in (at least) exponential time and usually produces a wastefully large model, if any exists. A more flexible and more practically applicable version, adopted in the present paper, is a so called \emph{incremental} (aka, `bottom-up') tableau building procedure. While in all known cases, the worst-case time complexity for maximal and incremental tableaux are the same, the crucial difference is that maximal tableaux \emph{always} require the amount of resources predicted by the theoretical worst-case time estimate, while incremental tableaux work on average much more efficiently\footnote{This claim can not be made mathematically precise due to the lack of an a priori probability distribution on formulae of a logic. The interested
reader may consult~\cite{GorSh08} for comparison of efficiency of
the two types of tableaux in the context of Alternating-time
temporal logic \ATL.}.  

\subsection{Related work and comparison}
\label{subsec:related}

The present work is part of a series of papers (\cite{GorSh08},\cite{GorSh08a},\cite{GorSh09},\cite{GorSh09a}) where we have embarked on the project of developing practically efficient yet intuitive and conceptually clear incremental-tableau-based satisfiability checking procedures for a range of multiagent logics.
This paper builds on the conference papers \cite{GorSh08a} and
\cite{GorSh09} by substantially extending, revising, and improving
them.
 
There are three inherent complications affecting the construction of a tableau procedure for the logic
\CMAELCD, arising respectively from the common knowledge
(fixpoint-definable operator), the distributed knowledge (with
associated epistemic relation being the intersection of the individual
knowledge epistemic relations), and the interactions between the
knowledge operators over different coalitions of agents.

Several tableau-based methods for satisfiability-checking for modal
logics with fixpoint-definable operators have been developed and published over the past 30 years, all going back to the tableau-based decision methods   developed for the Propositional Dynamic Logic \PDL\ in~\cite{Pratt80}, for the branching-time temporal logics \textbf{UB} in~\cite{BPM83} and \CTL\ in~\cite[Section 5]{EmHal85} and~\cite{Emerson90}. In terms of handling eventualities arising from the fixed-point operators our tableau method follows more closely on the incremental tableaux for the linear time temporal logic \LTL\ in~\cite{Wolper85} and for \CTL\ in~\cite[Section 7]{EmHal85}.

A particular complication arising in the tableau for \CMAELCD\ stems from the fact that the epistemic operators, being \system{S5} modalities, are symmetric, and thus the epistemic boxes have global effect on the model, too. This requires a special mechanism for propagating their effect backwards when occurring in
states of the tableau. In the present paper we have chosen to implement such mechanism by using \emph{analytic cut rules}, going back to Smullyan
\cite{Smullyan68} and Fitting \cite{Fitting83}, see also \cite{Gore98} and  \cite{Nguyen00}.  
More recently, tableaux with analytic cut rules for modal logics with symmetric relations have been developed in \cite{GoreN07TABLEAUX}, \cite{GoreN07CLIMA}, \cite{Dunin-KepliczNS11}. 

We note that there is a natural tradeoff between conceptual clarity and simplicity of (tableau-based) decision procedures on the one hand, and their technical sophistication and optimality on the other hand. We emphasize that the main objective of developing the tableau procedure presented here is the conceptual clarity, intuitiveness, and ease of  implementation, rather than practical optimality. While being optimal in terms of worst-case time complexity and incorporating some new and non-trivial optimizing features (such as restricted applications of cut rules) this procedure 
is amenable to various improvements and further optimizations. Most important known such optimizations are \emph{on-the-fly} techniques for elimination of bad states and \emph{one-pass} tableau methods  developed for some related logics in \cite{Schw98}, \cite{AGF07} and \emph{cut-free} versions of tableau as in \cite{AGF07} for \MAELC, \cite{GoreWidmann10} for PDL with converse operators, \cite{Nguyen11} for the description logic {SHI} and of sequent calculi, in~\cite{JKS07} for \MAELC\ and in \cite{BrunnlerLange08} for LTL and CTL.  
We discuss briefly the possible modifications of our procedure, implementing such optimizing techniques in Section \ref{sec:complexity}. 

Here is a summary (in a roughly chronological order) of the more closely related previous work, besides our own, on tableau-based decision procedures for multiagent epistemic logics with common and/or distributed knowledge:

\begin{itemize}
\item the maximal tableaux for \MAELC, presented in~\cite{HM92};
\item the semantic construction used in~\cite[Appendix A1]{FHV92} to
  prove completeness of an axiomatic system for \MAELD;
\item the proof of decidability of \MAELCD\ based on finite model
  property via filtration in~\cite{vdHM97};
\item the maximal tableau-like decision procedure for \ATL, presented
  in~\cite{WLWW06} and extended to \ATEL\ in~\cite{Walther05};
\item the exponential-time tableau-based procedure developed in
  \cite{DVD07} for testing satisfiability in the BDI logic, that has some 
  common features with \CMAELC;
\item  the optimized cut-free single-pass tableaux for the multi-agent
logic of common knowledge \MAELC, in \cite{AGF07}.
on tableaux for multiagent logics using global caching and analytic cuts in  \cite{Dunin-KepliczNS11}. 
\end{itemize}

\subsection{Structure of the paper} 
In Section~\ref{sec:logics}, we
introduce the syntax and semantics of the logic \CMAELCD.  In
Section~\ref{sec:HintikkaStructures}, we introduce Hintikka structures
for \CMAELCD\ and show that Hintikka structures are equivalent to Kripke models with
respect to satisfiability of formulae.  Then, in
Section~\ref{sec:tableaux}, we develop the tableau procedures checking
for satisfiability of formulae of \CMAELCD.
In Section~\ref{sec:scc}, we prove the correctness of our procedure in Section~\ref{sec:complexity} we estimate its complexity, discuss it efficiency and indicate some possible technical improvements. We end with concluding remarks pointing out some directions for further development.

%%%%%%%%%%%%%%%%%%%%%%%%%%%%%%%%%%%%%%%%%%%%%%
%%%%%%%%%%%%%%%%%%%%%%%%%%%%%%%%%%%%%%%%%%%%%%
\section{Syntax and semantics}
\label{sec:logics}

\subsection{Syntax of \CMAELCD}
The language of \CMAELCD\ contains a fixed, at most countable, set
\ap\ of \de{atomic propositions}, typically denoted by $p, q, r,
\ldots$; a finite, non-empty set $\agents$ of (names for)
\de{agents}\footnote{The notion of agent used in the present paper is
  an abstract one; in the context of distributed systems, for example,
  agents can be thought of as processes making up the system; in the
  context of multiagent systems, they can be thought of as independent
  software components of the system.}, typically denoted by $a, b,
\ldots$, while sets of agents, called \emph{coalitions}, will be
usually denoted by $A, B, \ldots$; a sufficient repertoire of the
Boolean connectives, say $\neg$ (``not'') and $\con$ (``and''); and,
for every non-empty coalition $A$, the epistemic operators
$\distrib{A}$ (``\emph{it is distributed knowledge among $A$ that
  \ldots}'') and $\commonk{A}$ (``\emph{it is common knowledge among
  $A$ that \ldots}'').
The formulae of \CMAELCD\ are thus defined by the following BNF
expression:
\[\vp := p \mid \neg \vp \mid (\vp_1 \con \vp_2) \mid \distrib{A}
\vp \mid \commonk{A} \vp,\] where $p$ ranges over $\ap$ and $A$ ranges
over the set $\powerne{\agents}$ of non-empty subsets of $\agents$.  
The other Boolean connectives can be defined as usual. We denote formulae of \CMAELCD\ 
by $\vp, \psi, \chi, \ldots$ and omit parentheses in formulae whenever it does not result in ambiguity.

The distributed knowledge operator $\distrib{A} \vp$ intuitively means
that an ``$A$-superagent'', who knows everything that any of the
agents in $A$ knows, can obtain $\vp$ as a logical consequence of
their knowledge. For example, if agent $a$ knows that $\psi$ and agent
$b$ knows that $\psi \imp \chi$, then $\distrib{\{a,b\}} \chi$ is true
even though neither $a$ nor $b$ knows $\chi$. The operators of
individual knowledge $\knows{a} \vp$ (``\emph{the agent $a$ knows that
  $\vp$}''), for $a \in \agents$, can be defined as $\distrib{\set{a}}
\vp$, henceforth simply written $\distrib{a} \vp$. Then, we define
$\knows{A} \vp := \bigwedge_{a \in A} \distrib{a} \vp$.

The common knowledge operator $\commonk{A} \vp$ intuitively means that
$\vp$ is ``public knowledge'' among $A$, i.e., that every agent in $A$
knows that $\vp$ and knows that every agent in $A$ knows that $\vp$,
etc.  Formulae of the form $\neg \commonk{A} \vp$ are referred to as
\de{(epistemic) eventualities}, for the reasons given later on.

\subsection{Coalitional multiagent epistemic models}
Formulae of \CMAELCD\ are interpreted in \emph{coalitional multiagent
  epistemic models}. In order to define those, we first need to
introduce \emph{coalitional multiagent epistemic structures} and
\emph{frames}.
  
\begin{definition}
 \label{def:cmaes}
 A \de{coalitional multiagent epistemic structure} (CMAES) is a tuple 
 \[
 \kframe{G}~=~(\agents, S, \set{\rel{R}^D_A}_{A
   \in \powerne{\agents}}, \set{\rel{R}^C_A}_{A \in
   \powerne{\agents}})
\]   
   where
 \begin{enumerate}
 \item $\agents$ is a finite, non-empty set of
   \de{agents}\footnote{Notice that we use the same symbol,
     ``$\agents$'', both for the set of \emph{names of agents} in the
     language and for the set of \emph{agents} in CMAES's.  It will
     always be clear from the context which set we refer to.};
 \item $S \ne \emptyset$ is a set of \de{states};
 \item for every $A \in
   \powerne{\agents}$, $\rel{R}^D_A$ is a binary relation on $S$;
 \item for every $A \in \powerne{\agents}$, $\rel{R}^C_A$ is the
   reflexive, transitive closure of $\bigunion_{B\subseteq A}\,
   \rel{R}^D_{B}$.
 \end{enumerate}
\end{definition}

\begin{definition}
 \label{def:cmaef}
 A \de{coalitional multiagent epistemic frame} (CMAEF) is a CMAES 
\[
\kframe{F} = (\agents, S, \set{\rel{R}^D_A}_{A \in
   \powerne{\agents}}, \set{\rel{R}^C_A}_{A \in \powerne{\agents}}),
\]
 where each $\rel{R}^D_A$ is an equivalence relation satisfying the
 following condition:
 \begin{center}
   \(
   \begin{array}{lcl}
     (\dag) & & \rel{R}^D_A = \biginter_{a \in A} \rel{R}^D_{a}
   \end{array}
   \)
 \end{center}
 (Here, and further, we write $\rel{R}^D_{a}$ instead of $\rel{R}^D_{\set{a}}$, where $a \in \agents$.)
 
 If condition (\dag) above is replaced by the following, weaker one:
  \begin{center}
   \(
   \begin{array}{lcl}
     (\dag \dag) & & \rel{R}^D_A \subseteq \rel{R}^D_B \mbox{ whenever } B \subseteq A,

   \end{array}
   \)
 \end{center}
 then \kframe{F} is a \de{coalitional multiagent epistemic pseudo-frame} (pseudo-CMAEF).
\end{definition}

Note that in every (pseudo-)CMAEF $\rel{R}^D_A \subseteq \biginter_{a
  \in A} \rel{R}^D_{a}$, and hence $\bigunion_{B\subseteq
  A} \rel{R}^D_{B} = \bigunion_{a \in A} \rel{R}^D_{a}$. 
  Hence, condition 4 of
Definition~\ref{def:cmaes} in (pseudo-) CMAEFs is equivalent to
requiring that $\rel{R}^C_A$ is the transitive closure of
$\bigunion_{a \in A} \rel{R}^D_{a}$, for every $A \in \powerne{\agents}$.
Also, note that each $\rel{R}^C_A$ in a (pseudo-)CMAEF is an equivalence
relation.

\begin{definition}
 \label{def:cmaem}
 A \de{coalitional multiagent epistemic model} (CMAEM) is a tuple
 $\mmodel{M} = (\kframe{F}, \ap, L)$, where $\kframe{F}$ is a CMAEF with a set of states $S$, $\ap$ is a set of atomic propositions, and $L: S \mapsto \power{\ap}$
 is a \de{labeling function}, assigning to every state $s$ the set
 $L(s)$ of atomic propositions true at $s$.

 If $\kframe{F}$ is a pseudo-CMAEF, then $\mmodel{M} = (\kframe{F},
 \ap, L)$ is a \de{multiagent coalitional pseudo-model}
 (pseudo-CMAEM).
\end{definition}

The notion of truth, or satisfaction, of a \CMAELCD-formula at a state of a 
 (pseudo-)CMAEM is defined in the standard Kripke semantics style.  In particular: 
\begin{itemize}
\item \sat{M}{s}{\distrib{A} \vp} iff $(s, t) \in \rel{R}^D_A$ implies
 \sat{M}{t}{\vp};
\item \sat{M}{s}{\commonk{A} \vp} iff $(s, t) \in \rel{R}^C_A$ implies
 \sat{M}{t}{\vp}.
\end{itemize}

\begin{definition}
  Given a (pseudo-)CMAEM $\mmodel{M}$, a \CMAELCD-formula $\vp$ is \de{satisfiable} in $\mmodel{M}$ if \sat{M}{s}{\vp} holds for some $s \in \mmodel{M}$; $\vp$ is \de{valid} in $\mmodel{M}$ if \sat{M}{s}{\vp} holds for every $s \in \mmodel{M}$.  
  
A formula $\vp$ is  \de{satisfiable} if it is satisfiable in some CMAEM; it is \de{valid}, denoted $\Vdash \vp$, if it is valid in every CMAEM. 
\end{definition}

The satisfaction condition for the operator $\commonk{A}$ can be
re-stated in terms of reachability.  Let $\mmodel{M}$ be a
(pseudo-)CMAEM with state space $S$ and let $s, t \in S$.  We say that
$t$ is \de{$A$-reachable from $s$} if either $s = t$ or, for some $n
\geq 1$, there exists a sequence $s = s_0 , s_1, \ldots, s_{n-1}, s_n
= t$ of elements of $S$ such that, for every $0 \leq i < n$, there
exists $a_i \in A$ such that $(s_i, s_{i+1}) \in \rel{R}^D_{a_i}$.  It is
then easy to see that the satisfaction condition for
$\commonk{A}$ is equivalent to the following one:
\begin{itemize}
\item \sat{M}{s}{\commonk{A} \vp} iff \sat{M}{t}{\vp} for every $t$ that is
 $A$-reachable from $s$.
\end{itemize}

The following claim be easily verified.
\begin{proposition}
\label{fixpointC}
$\Vdash  \commonk{A} \vp  \leftrightarrow (\vp \land \bigwedge_{a\in A} \distrib{a}\commonk{A} \vp)$. 
\end{proposition}

\textsc{Remark} If $\agents = \set{a}$, then $\distrib{a}
\vp \equivalence \commonk{a} \vp$ is valid for all $\vp$.  Thus, the
single-agent case is essentially trivialized and, therefore, we assume
throughout the remainder of the paper that the set $\agents$ of names
of agents in the language of \CMAELCD\ contains at least 2 agents.

%%%%%%%%%%%%%%%%%%%%%%%%%%%%%%%%%%%%%%%%%%%%%
%%%%%%%%%%%%%%%%%%%%%%%%%%%%%%%%%%%%%%%%%%%%%%
\section{Hintikka structures for  \CMAELCD}
\label{sec:HintikkaStructures}

We are ultimately interested in (constructive) satisfiability of
(finite sets of) formulae in models. However, the tableau procedure we
present in this paper checks for the existence of a more general kind
of semantic structure for the input formula, namely a \emph{Hintikka
  structure}.  In Section~\ref{subsec:HintikkaStructures}, we introduce Hintikka structures for
\CMAELCD.  In Section~\ref{subsec:Hintikka_equals_models} we show that satisfiability in
Hintikka structures is equivalent to satisfiability in models;
consequently, testing for satisfiability in a Hintikka structure can
replace testing for satisfiability in a model.

\subsection{Fully expanded sets and Hintikka structures}
\label{subsec:HintikkaStructures}

There are two fundamental differences between (pseudo-)models and
Hintikka structures for \CMAELCD, which make working with the latter
substantially easier than working directly with models. First, while
models specify the truth value of every formula of the language at
each state, Hintikka structures only do so for the formulae relevant
to the evaluation of a fixed formula $\theta$ (or, a finite set of
formulae $\Theta$) at a distinguished state.  Second, the relations in
(pseudo-) models have to satisfy certain conditions (see
Definition~\ref{def:cmaef}), while in Hintikka structures conditions
are only placed on the labels of states.  These labeling conditions
ensure, however, that every Hintikka structure generates, through the
constructions described in Section \ref{subsec:Hintikka_equals_models}, 
a pseudo-model so that
membership of formulae in the labels is compliant with the truth in
the resultant pseudo-model.  We then show how to convert a
pseudo-model into a bona fide model in a ``truth-preserving'' way.

To describe Hintikka structures, we need the concept of fully expanded
set. Such sets contain all the formulae that have to be satisfied
locally at the state under consideration. We divide all the formulae
that are not elementary in the sense that their satisfaction at the
state does not imply satisfaction of any other formulae at the same
state (such as $p \in \ap$ or $\neg \distrib{A} \vp$) into
$\alpha$-formulas and $\beta$-formulas.  The former are formulae of a
conjunctive type, i.e. their truth implies the truth of all their
$\alpha$-components at the same state, while the latter are of a
disjunctive type: their truth implies the truth of at least one of
their $\beta$-components at the same state. Table
\ref{tab:alpha_beta_components} shows the $\alpha$- and
$\beta$-formulas of $\CMAELCD$ together with their $\alpha$- and
$\beta$-components.  The following claims are straightforward, the
cases of common knowledge using Proposition \ref{fixpointC}.

\begin{table}\label{tab:alpha_beta_components}
\begin{tabular}[t]{|c|c|}
\hline
$\alpha$-formula & $\alpha$-components\\
\hline
$\neg \neg \vp$&$\{\vp\}$\\
$\vp \con \psi$&$\{\vp, \psi\}$\\
$\distrib{A} \vp$&$\{\distrib{A}\vp, \vp\}$\\
$\commonk{A} \vp$&$\{\vp\}\cup \crh{\distrib{a}\commonk{A}\vp}{a\in A}$\\
\hline
\end{tabular}
\begin{tabular}[t]{|c|c|}
\hline
$\beta$-formula&$\beta$-components\\
\hline
$\neg (\vp \con \psi)$ &$\{\neg \vp,\neg \psi\}$\\
$\neg \commonk{A} \vp$&$\{\neg\vp\}\cup\crh{\neg \distrib{a}\commonk{A}\vp}{a\in A}$\\
\hline
\end{tabular} \\ ~ \\
\caption{$\alpha$- and $\beta$-formulas of $\CMAELCD$ with their respective components} 
\end{table}

\begin{lemma}
\label{lem:Decomp}
\begin{description}
\item[1.] Every $\alpha$-formula is equivalent to the conjunction of
  its $\alpha$-components.
\item[2.] Every $\beta$-formula is equivalent to the disjunction of
  its $\beta$-components.
\end{description}
\end{lemma}

\begin{definition}
\label{def:Clo}

The  \emph{closure of the formula} $\vp$ is the smallest set of formulae $\cl{\vp}$ such that:
  \begin{enumerate}
  \item $\vp \in \cl{\vp}$;
  \item $\cl{\vp}$ is closed with respect to $\alpha$- and
    $\beta$-components of all $\alpha$- and $\beta$-formulae,
    respectively;
  \item for any formula $\psi$ and coalition $A$, if $\lnot
    \distrib{A}\psi \in \cl{\vp}$ then $\lnot \psi \in \cl{\vp}$.
\end{enumerate}
\end{definition}

\begin{definition}
\label{def:Clo2}
For any set of formulae $\De$ we define $\cl{\De} := \bigcup \{
\cl{\vp} \mid \vp \in \De\}$. A set of formulae $\De$ is \emph{closed}
if $\De = \cl{\De}$.
\end{definition}

\begin{remark}
  Intuitively, the closure of a set of formulae $\Ga$ consists of all
  formulae that may appear in the tableau whose input is the set of
  formulae $\Ga$.
\end{remark}

\begin{definition}
\label{def:PatInc}
A set of formulae is \emph{patently inconsistent} if it contains 
a contradictory pair of formulae $\vp$ and $\neg \vp$.
\end{definition}

\begin{definition}
  \label{def:cmaelcd_downward_saturated}
  A set $\De$ of \CMAELCD-formulae is \de{fully expanded} if it
  satisfies the following conditions:
\begin{itemize}
	\item $\De$ is not \de{patently inconsistent}; 
	\item if $\vp$ is an $\alpha$-formula and $\vp\in \De$, then
          all $\alpha$-components of $\vp$ are in $\De$.
	\item if $\vp$ is a $\beta$-formula and $\vp\in \De$, then at
          least one $\beta$-component of $\vp$ is in $\De$.
\end{itemize}
\end{definition}

Intuitively, a non-patently inconsistent set is fully expanded if it
is closed under applications of all \emph{local} (pertaining to the
same state of a structure) formula decomposition rules.

\begin{definition}
  \label{def:FullExp}
 
  The procedure $\fexp$ applies to a set of formulae $\Ga$ and
  produces a (possibly empty) family of sets $\fe{\Ga}$, called the
  \emph{family of full expansions of $\Ga$}, obtained as follows:
  start with the singleton family $\{\Ga\}$; if $\Ga$ is patently
  inconsistent, halt and return $\fe{\Ga} = \emptyset$; otherwise
  repeatedly apply, until saturation, the following \emph{set
    replacement operations}, each time to a non-deterministically
  chosen set $\Phi$ from the current family of sets $\mathcal{F}$ and
  a formula $\varphi \in \Phi$; though, we prioritize the
  eventualities in $\Ga$ so that these formulae are processed first:

 \begin{enumerate}
\item
If $\varphi$ is an $\alpha$-formula with $\alpha$-components $\varphi_1$ and $\varphi_2$, then replace $\Phi$ by $\Phi \cup \{\varphi_1,\varphi_2\}$.
\item
If $\varphi$ is a $\beta$-formula such that none of its  $\beta$-components is in $\Phi$, then replace $\Phi$ with the family of extensions 
\[
\{\Phi \cup \{\psi\} \mid \psi \mbox{ is a $\beta$-component of } \varphi \}
\]
\item
If $\varphi = \neg \commonk{A} \psi$ and $\neg \psi \notin \Phi$,  
but some of the other $\beta$-components of $\varphi$ is in $\Phi$, then add to $\mathcal{F}$ the set $\Phi \cup \{\neg \psi\}$ 
 \end{enumerate}
 
 The following proviso applies to the procedure above:  if a patently inconsistent set is added to $\mathcal{F}$ as a result of such application, it is removed immediately thereafter. 
 
Saturation occurs when no application of a set replacement operation can change the current family $\mathcal{F}$. At that stage, the family $\fe{\Ga}$ of sets of formulae is produced and returned. Reaching a stage of saturation is guaranteed to occur because all sets of formulae produced during the procedure $\fexp$ are subsets of the finite set $\cl{\Ga}$.
\end{definition}

Notice that the procedure $\fexp$ allows adding not more than one $\beta$-component of a formula $\varphi = \neg \commonk{A} \psi$ to the initial set, besides $\neg \psi$.

In what follows, we will need the following proposition.

\begin{proposition}
  \label{lem:full_expansion}
For any finite set of formulae $\Ga$:
\[\Vdash \bigwedge \Ga \leftrightarrow \bigvee \left\{\bigwedge \De \mid \De \in
\fe{\Ga}\right\}.\]
\end{proposition}

\begin{proof}
By Lemma~\ref{lem:Decomp}, 
  every set replacement operation applied to a family $\mathcal{F}$ preserves the formula
  $\bigvee \{\bigwedge \De \mid \De \in \fe{\Ga}\}$ up to logical
  equivalence. At the beginning, that formula is $\bigwedge \Ga$, hence 
  the claim follows.
\end{proof}

We now define Hintikka structures for \CMAELCD:

\begin{definition}
 \label{def:cmaehs}
 A \de{coalitional multiagent epistemic Hintikka structure} (CMAEHS)
 is a tuple 
 \[(\agents, S, \set{\rel{R}^D_A}_{A \in \powerne{\agents}},
 \set{\rel{R}^C_A}_{A \in \powerne{\agents}}, \ap, H)\] such that:
 \begin{itemize}
 \item $(\agents, S, \set{\rel{R}^D_A}_{A \in \powerne{\agents}},
   \set{\rel{R}^C_A}_{A \in \powerne{\agents}})$ is a CMAES (recall
   Definition~\ref{def:cmaes});
 \item \ap\ is a set of atomic propositions;
 \item $H$ is a labeling of the elements of $S$ with sets of
   \CMAELCD-formulae that satisfy the following constraints, for every
   $s, s' \in S$:
   
   \begin{enumerate}[label={{CH\arabic*}}, ref=(CH\arabic*)]
   \item\label{it:CHDS} $H(s)$ is fully expanded;
   \item\label{it:CHnegDK} If $\neg \distrib{A} \vp \in H(s)$, then $(s, t) \in
     \rel{R}^D_A$ and $\neg \vp \in H(t)$, for some $t\in S$;
   \item\label{it:CHDK} If $(s, s') \in \rel{R}^D_A$, then $\distrib{B} \vp \in
     H(s)$ iff $\distrib{B} \vp \in H(s')$, for every $B \subseteq A$;
   \item\label{it:CHnegCK} If $\neg \commonk{A} \vp \in H(s)$, then $(s, t) \in
     \rel{R}^C_A$ and $\neg \vp \in H(t)$, for some $t \in S$.
   \end{enumerate}
 \end{itemize}
\end{definition}

\begin{definition}
  \label{def:Hintikka_sat}
Let $\hintikka{H}$ be a CMAEHS with state space $S$. A  \CMAELCD-formula $\theta$  is \de{satisfiable} in \hintikka{H} if $\theta
  \in H(s)$, for some $s \in S$.  Likewise, a set of  \CMAELCD-formulae $\Theta$ is satisfiable in \hintikka{H} if $\Theta \subseteq H(s)$, for some 
  $s \in S$.
\end{definition}

%%%%%%%%%%%%%%%%%%%%%%%%%%%%%%%%%%%%%%%%%%%%
%%%%%%%%%%%%%%%%%%%%%%%%%%%%%%%%%%%%%%%%%%%%%
\subsection{Equivalence of Hintikka structures and models for \CMAELCD}
\label{subsec:Hintikka_equals_models}

Here we show that satisfiability in Hintikka structures is
equivalent to satisfiability in models. For brevity, we only deal with single formulae; the extension to finite sets of formulae is straightforward.  The main complications in the proofs below arise due to the presence of distributed knowledge operators in the language
of a logic.  

Here we will prove that a \CMAELCD-formula $\theta$ is satisfiable in
a CMAEM iff it is satisfiable in a CMAEHS.  First, we show that
satisfiability in a CMAEM implies satisfiability in a CMAEHS.  Then, we
show that satisfiability in a CMAEHS implies satisfiability in a
pseudo-CMAEM, which in turn implies satisfiability in a CMAEM.

That satisfiability in a CMAEM implies satisfiability in a CMAEHS is
almost immediate. Given a CMAEM $\mmodel{M}$
with a set of states $S$, define the \emph{extended labeling
  function} $L^+_{\mmodel{M}}$ from $S$ to the power-set of
\CMAELCD-formulae as follows: $L^+_{\mmodel{M}}(s) =
\crh{\vp}{\sat{M}{s}{\vp}}$.  It is then routine to check the
following.

\begin{lemma}
  \label{lem:models_to_Hintikka}
  Let $\mmodel{M} = (\agents, S, \set{\rel{R}^D_A}_{A \in
    \powerne{\agents}}, \set{\rel{R}^C_A}_{A \in \powerne{\agents}},
  \ap, L)$ be a CMAEM satisfying $\theta$ and let $L^+_{\mmodel{M}}$ be
  the extended labeling on \mmodel{M}.  Then, $(\agents, S,
  \set{\rel{R}^D_A}_{A \in \powerne{\agents}},$ $\set{\rel{R}^C_A}_{A
    \in \powerne{\agents}}, \ap, L^+_{\mmodel{M}})$ is a CMAEHS
  satisfying $\theta$.  Therefore, satisfiability in a CMAEM implies
  satisfiability in a CMAEHS.
\end{lemma}

For the converse direction we need two steps, done in  
Lemma~\ref{lem:Hintikka_to_psmodels} and Lemma~\ref{lem:psmodels_to_models}.

\begin{lemma}
 \label{lem:Hintikka_to_psmodels}
 Let $\theta$ be a \CMAELCD-formula satisfiable in a CMAEHS.  Then,
 $\theta$ is satisfiable in a pseudo-CMAEM.
\end{lemma}

\begin{proof}
 Let $\hintikka{H} = (\agents, S, \set{\rel{R}^D_A}_{A \in
   \powerne{\agents}}, \set{\rel{R}^C_A}_{A \in \powerne{\agents}},
 \ap, H)$ be an CMAEHS for $\theta$.  We construct a pseudo-CMAEM
 \mmodel{M'} satisfying $\theta$ out of \hintikka{H} as follows.

 First, for every $A \in \powerne{\agents}$, let $\rel{R}'^D_A$ be the
 reflexive, symmetric, and transitive closure of $\bigunion_{A
   \subseteq B} \rel{R}^D_{B}$ and let $\rel{R}'^C_A$ be the
 transitive closure of $\bigunion_{a \in A} \rel{R}'^D_a$. Thus, both 
 $\rel{R}'^D_A$ and  $\rel{R}'^C_A$ are equivalence relations and 
 $\rel{R}^D_A \subseteq \rel{R}'^D_A$ and $\rel{R}^C_A \subseteq
 \rel{R}'^C_A$, for every $A \in \powerne{\agents}$. Second, let $L(s)
 = H(s) \inter \ap$, for every $s \in S$.  It is then immediate to check
 that $B \subseteq A$ implies $\rel{R}'^D_A \subseteq \rel{R}'^D_B$,
 and hence, $\mmodel{M}' = (\agents, S, \set{\rel{R}'^D_A}_{A \in
   \powerne{\agents}}, \set{\rel{R}'^C_A}_{A \in \powerne{\agents}},
 L)$ is a pseudo-CMAEM.

Basically this construction relabels the edges of a Hintikka structure such that if a directed edge is labelled with a coalition $A$, it is made bidirectional and is further labelled with all coalitions that are subsets of $A$. Hereafter the relation is then made transitive and reflexive. The labels of the states are reduced to only containing (positive) atoms. 
Figure \ref{fig:HS_to_pseudo} illustrates the process of  transforming the Hintikka structure on the left into the pseudo-model on the right. 
\begin{figure}
\begin{equation*}
\small{\xymatrix{
\{\neg \di{a}\neg p,\di{b}q,q\}\ar[d]_{\{a\},\{b\}}\ar[r]^(0.45){\{a,b\}}&\{\di{b}q,q,\neg\co{\{a,b\}}r,\neg r\}\ar@(dr,dl)^{\{a,b\}}\\
\{\neg\neg p,p,\di{b}q,q\}
}}\quad \leadsto\quad 
\small{
\xymatrix{
\{q\}\ar@(dl,ul)^{\tiny{\begin{tabular}{@{}>{$}c<{$}@{}}\set{a,b},\\ \set{a},\set{b}\end{tabular}}}\ar@{<->}[d]_{\tiny{\set{a},\set{b}}}\ar@{<->}[r]^{\tiny{\begin{tabular}{@{}>{$}c<{$}@{}}\set{a,b},\\ \set{a},\set{b}\end{tabular}}}&\{q\}\ar@(dr,ur)_{\tiny{\begin{tabular}{@{}>{$}c<{$}@{}}\set{a,b},\\ \set{a},\set{b}\end{tabular}}}\\
\{p.q\}\ar@(dr,dl)^{\tiny{\set{a,b},\set{a},\set{b}}}\ar@{<->}[ur]_{\tiny{\set{a},\set{b}}}
}
}
\end{equation*}
\caption{Example on transforming a Hintikka structure to a pseudo-model using the construction from the proof of Lemma~\ref{lem:Hintikka_to_psmodels}}
\label{fig:HS_to_pseudo} 
\end{figure}
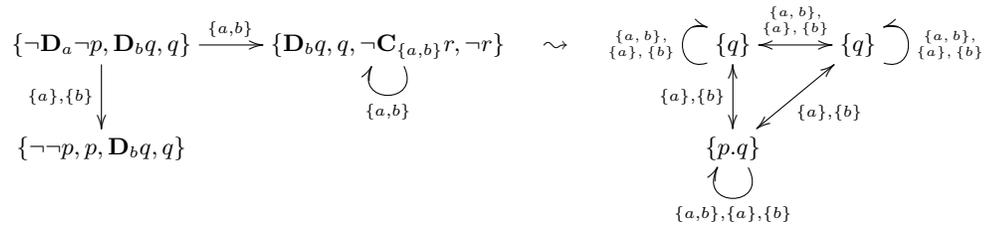

 To complete the proof of the lemma, we show, by induction on the
 structure of the  formulae in $\cl{\theta}$ that, for every $s \in S$ and every
 formula $\chi$, the following hold:
 
 \begin{center}
 \(
 \begin{array}{ll}
   (i) & \chi \in H(s) \text{ implies } \sat{M'}{s}{\chi}; \\
   (ii) & \neg \chi \in H(s) \text{ implies } \sat{M'}{s}{\neg \chi}.
 \end{array}
 \)  
 \end{center}

 \noindent The statement of the lemma then follows.
 
 Let $\chi$ be some $p \in \ap$.  Then, $p \in H(s)$ implies $p \in
 L(s)$ and, thus, \sat{M'}{s}{p}; if, on the other hand, $\neg p \in
 H(s)$, then due to \ref{it:CHDS}, $p \notin H(s)$ and thus $p \notin L(s)$;
 hence, \sat{M'}{s}{\neg p}.

 Assume that the claim holds for all subformulae of $\chi$; then, we have
 to prove that it holds for $\chi$, as well.

 Suppose that $\chi$ is $\neg \vp$.  If $\neg \vp \in H(s)$, then the
 inductive hypothesis immediately gives us $\sat{M'}{s}{\neg \vp}$;
 if, on the other hand, $\neg \neg \vp \in H(s)$, then by virtue of
 \ref{it:CHDS}, $\vp \in H(s)$ and hence, by inductive hypothesis,
 $\sat{M'}{s}{\vp}$ and thus $\sat{M'}{s}{\neg \neg \vp}$.

 The case of $\chi = \vp \con \psi$ is straightforward, using \ref{it:CHDS}.

 Suppose that $\chi$ is $\distrib{A} \vp$.  Assume, first, that
 $\distrib{A} \vp \in H(s)$.  In view of the inductive hypothesis, it
 suffices to show that $(s, t) \in \rel{R}'^D_A$ implies $\vp \in
 H(t)$.  So, assume that $(s, t) \in \rel{R}'^D_A$.  There are two
 cases to consider.  If $s = t$, then the conclusion immediately
 follows from \ref{it:CHDS}.  If, on the other hand, $s \ne t$, then there
 exists an undirected path between $s$ and $t$ along the relations of
 the form $\rel{R}^D_{B}$, where each $B$ is a superset of $A$.  Then,
 in view of \ref{it:CHDK}, $\distrib{A} \vp \in H(t)$; hence, by \ref{it:CHDS}, $\vp
 \in H(t)$, as desired.

 Assume, next, that $\neg \distrib{A} \vp \in H(s)$.  In view of the
 inductive hypothesis, it suffices to show that there exists $t \in S$
 such that $(s, t) \in \rel{R}'^D_A$ and $\neg \vp \in H(t)$.  By
 \ref{it:CHnegDK}, there exists $t \in S$ such that $(s, t) \in \rel{R}^D_A$ and
 $\neg \vp \in H(t)$.  As $\rel{R}^D_A \subseteq \rel{R}'^D_A$, the
 desired conclusion follows.

 Suppose now that $\chi$ is $\commonk{A} \vp$.  Assume that $\commonk{A}
 \vp \in H(s)$.  In view of the inductive hypothesis, it suffices to
 show that if $t$ is $A$-reachable from $s$ in $\mmodel{M}'$, then
 $\vp \in H(t)$. So, assume that either $s = t$ or, for some $n \geq
 1$, there exists a sequence of states $s = s_0, s_1, \ldots, s_{n-1},
 s_n = t$ such that, for every $0 \leq i < n$, there exists $a_i \in
 A$ such that $(s_i, s_{i+1}) \in \rel{R}'^D_{a_i}$. In the former
 case, the desired conclusion follows from \ref{it:CHDS}. In the latter, we
 can show by induction on $i$, for $0 \leq i < n$, using \ref{it:CHDK}
 and \ref{it:CHDS}, that $\distrib{a_i} \commonk{A} \vp \in H(s_i)$.  Then, in particular, $\distrib{a_{n-1}} \commonk{A} \vp \in H(s_{n-1})$, and again, by \ref{it:CHDK}, $\distrib{a_{n-1}}\commonk{A}\vp\in H(t)$ and thus by \ref{it:CHDS}, $\commonk{A}\vp\in H(t)$ and $\vp \in H(t)$.

 Assume, on the other hand, that $\neg \commonk{A} \vp \notin H(s)$.
 Then, the desired conclusion follows from \ref{it:CHnegCK}, the inclusion
 $\rel{R}^C_A \subseteq \rel{R}'^C_A$, and the inductive hypothesis.
 
\end{proof}

We now prove that satisfiability in a pseudo-CMAEM implies
satisfiability in a CMAEM. To that end, we use a modification of the
construction from~\cite[Appendix A1]{FHV92} to show that if $\theta$
is satisfiable in a pseudo-CMAEM, then it is satisfiable in a
``tree-like'' pseudo-CMAEM that actually turns out to be a bona fide
CMAEM.
 To present the proof, we need some preliminary definitions.

\begin{definition}
  \label{def:maximal_paths}
  Let $\mmodel{M} = (\agents, S, \set{\rel{R}^D_A}_{A \in
    \powerne{\agents}}, \set{\rel{R}^C_A}_{A \in \powerne{\agents}}, \ap,
  L)$ be a \mbox{(pseudo-)} CMAEM and let $s, t \in S$.  A \de{maximal
    path from $s$ to $t$} in \mmodel{M} is a sequence $s_0, A_0, s_1,
  A_1,\ldots, s_{n-1}, A_{n-1}, s_n$ where $s=s_0$ and $t=s_n$, such
  that $n=0$ and $s=t$ or, for every $0 \leq i < n$, $(s_i, s_{i+1}) \in
  \rel{R}^D_{A_i}$, but $(s_i, s_{i+1}) \in \rel{R}^D_{B}$ does not
  hold for any $B$ with $A_i \subset B \subseteq \agents$.  A segment
  $\rho'$ of a maximal path $\rho$ starting and ending with a state is
  a \de{sub-path} of $\rho$.
\end{definition}

Notice that, in general, there might be several maximal paths between
a pair of states.

For a path $\tau = s_0, A_0, s_1, \ldots, s_{n-1}, A_{n-1}, s_n$, we
denote by $\tau_{|i}$ the sub-path of $\tau$ starting in $s_0$ and
ending in $s_i$, i.e. $\tau_{|i} = s_0, A_0, s_1,\ldots, A_{i-1},
s_{i}$ and by $|\tau|$ the length of $\tau$, i.e. $n$. We denote the
last element of a path $\tau$, which is a state, by $l(\tau)$ and
the second last element of $\tau$, which is a coalition, by
$sl(\tau)$.

\begin{lemma}
 \label{lem:psmodels_to_models}
 Let $\theta$ be a \CMAELCD-formula satisfiable in a pseudo-CMAEM;
 then, $\theta$ is satisfiable in a CMAEM.
\end{lemma}

\begin{proof}
  Suppose that $\theta$ is satisfied in a pseudo-CMAEM $\mmodel{M}$ at
  state $s$. Let $\mmodel{M}_s = (\agents, S, \set{\rel{R}^D_A}_{A \in
    \powerne{\agents}}, \set{\rel{R}^C_A}_{A \in \powerne{\agents}},
  \ap, L)$ be the submodel of \mmodel{M} generated by $s$.  Then,
  $\mmodel{M}_s, s \Vdash \theta$ since $\mmodel{M}_s$ and
  $\mmodel{M}$ are locally bisimilar at $s$. Next, we unravel
  $\mmodel{M}_s$ into a model $\mmodel{M}^* = (\agents, S^*,
  \set{\rel{R}^*{}^D_A}_{A \in \powerne{\agents}},
  \set{\rel{R}^*{}^C_A}_{A \in \powerne{\agents}}, \ap, L^*)$, as
  follows.

  First, call a maximal path $\rho$ in $\mmodel{M}_s$ an $s$-max-path
  if the first component of $\rho$ is $s$, and let $S^{*}$ be the set
  of all $s$-max-paths in $\mmodel{M}_s$.  Notice that $s$ by itself
  is an $s$-max-path with $l(s) = s$.
  
  For every $A\in \powerne{\agents}$, let 
  \[
  \rel{R}'^D_A = \crh{(\rho, \tau)} {\rho, \tau \in S^*,\: \tau_{|\,
      |\tau|-1} = \rho \text{ and } sl(\tau)\supseteq A},
  \]
  i.e.  $(\rho, \tau) \in \rel{R}'^D_A$ if $\tau$ extends $\rho$ with
  one step labelled by a coalition containing $A$.
  Next, let $\rel{R}^*{}^D_A$ be a reflexive, symmetric, and
  transitive closure of $\rel{R}'^D_A$.  Notice that $(\rho,\tau)\in
  \rel{R}^*{}^D_A$ holds for two distinct paths $\rho$ and $\tau$ iff
  there exists a sequence $\rho_0, \ldots, \rho_n\in S^*$ with
  $\rho=\rho_0$ and $\tau=\rho_n$ such that for all $i<n$, either
  $(\rho_i,\rho_{i+1})\in\rel{R}'^D_A$ or
  $(\rho_{i+1},\rho_{i})\in\rel{R}'^D_A$.
  It then follows that the following \emph{downward closure condition}
  holds:
    
    $$\textbf{(DC)} \text{ If }(\rho, \tau) \in
  \rel{R}^*{}^D_A \text{ and } B \subseteq A, \text{ then } (\rho,
  \tau) \in \rel{R}^*{}^D_B.$$ 

  The relations $\rel{R}^*{}^C_A$ are defined as in any CMAEF.  To
  complete the definition of $\mmodel{M}^*$, we put $L^*(\rho) =
  L(l(\rho))$, for every $\rho \in S^*$. Notice that $\mmodel{M}^*$ is
  \emph{tree-like} in the sense that the structure $(S^*,
  \set{\rel{R}'^D_A}_{A \in \powerne{\agents}})$ is a
  tree. 

By this construction we basically remove all `non-maximal' edges between two vertices from the part of the given pseudomodel that can be reached by the given state $s$. Then we build paths by starting in $s$ and then traversing the resulting graph via the edges. E.g., if we consider the pseudo-model $\mmodel{M}$ in Figure \ref{fig:HS_to_pseudo}, and we let the top-left-most state be $s$, then $\sat{M}{s}{\neg\di{a}\neg p\land \di{b}q}$. $S^*$ will in this case be all paths starting in $s$ and following  
the links in the graph.  
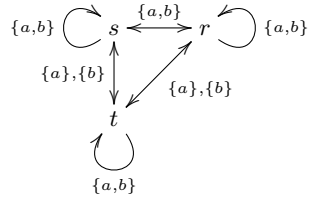
\[
\small{
\xymatrix{
s\ar@(dl,ul)^{\set{a,b}}\ar@{<->}[d]_{\set{a},\set{b}}\ar@{<->}[r]^{\set{a,b}}&r\ar@(dr,ur)_{\set{a,b}}\\
t\ar@(dr,dl)^{\set{a,b}}\ar@{<->}[ur]_{\set{a},\set{b}}
}
}
\]
I.e. $\rho = (s,\{a,b\},s,\{a\},t)$ and $\tau = (s,\{a,b\},r,\{b\},t)$ are in $S^*$, while $\rho'=(s,\{a\},s,\{b\},t)\notin S^*$. 

We have $(\rho,\tau)\notin \rel{R}^*{}^D_{a}$, $(\rho,\tau)\notin \rel{R}^*{}^D_{b}$ and $(\rho,\tau)\notin \rel{R}^*{}^D_{a,b}$. On the other hand, $(\tau, (s,\{a,b\},r,\{a,b\},s) )\in\rel{R}^*{}^D_{b}$.

In this example, $L^*(\rho)=L^*(\tau)\overset{\text{def.}}=L(t)=\{p,q\}$.

  It is clear from the construction, namely from \textbf{(DC)}, that
  $\mmodel{M}^*$ is a pseudo-CMAEM, and in the following, we will show
  that condition ($\dag$) of Definition \ref{def:cmaef} also holds, so that
  $\mmodel{M}^*$ is a CMAEM.

  First, we notice that, since $\mmodel{M}^*$ is tree-like, we have
  $(\rho, \tau)\in\rel{R}^*{}^D_A$ iff there exists $k \geq 0$, with
  $k\leq|\rho|$ and $k\leq|\tau|$, such that
  \begin{equation}\label{eq:condforin*}
\begin{gathered}
  \rho_{|k} = \tau_{|k}, \text{ and}\\
  \text{ for all } k<i\leq|\tau| \text{ and } k<j\leq|\rho|,
  A\subseteq sl(\tau_{|i}) \text{ and } A\subseteq sl(\rho_{|j}).
\end{gathered}
\end{equation}

(The situation is depicted in Figure \ref{fig:possiblepaths}.)
\begin{figure}[htbp]\centering 
\setlength{\unitlength}{8mm}
\begin{picture}(7.5, 2.5)
 \put(-1,1){\circle*{0.15}}
 \qbezier(-1, 1)(-0.5,1.5)(0,  1)
 \put(-1,1.3){\small{$s$}}
 \put(0,1){\circle{0}}
 \qbezier(0, 1)(0.5,0.5)(1,  1)
 \put(1,1){\circle{0}}
 \put(1.3,1){\small{$\ldots$}}
 \put(2,1){\circle{0}}
 \qbezier(2, 1)(2.5,1.5)(3,  1)
 \put(3,1){\circle*{0.15}}
 \put(3.3,1){\small{$l(\rho_{|k})=l(\tau_{|k})$}}

 \qbezier(3, 1)(4,2)(4,  2)
 \put(3.1,1.7){\tiny{$\supseteq A$}}
 \put(4,2){\circle{0}}

 \put(4.3,2){\small{$\ldots$}}

 \put(5,2){\circle{0}}
 \qbezier(5, 2)(5.5,2.5)(6,  2)
 \put(5.2,2.4){\tiny{$\supseteq A$}}
 \put(6,2){\circle{0}}
 \qbezier(6, 2)(6.5,1.5)(7,  2)
 \put(6.2,1.9){\tiny{$\supseteq A$}}
 \put(7,2){\circle{0}}
 \qbezier(7, 2)(7.5,2.5)(8,  2)
 \put(7.2,2.4){\tiny{$\supseteq A$}}
 \put(8,2){\circle*{0.15}}
 \put(8.3,2){\small{$l(\tau)$}}

 \qbezier(3, 1)(3,1)(4,  0)
 \put(3.0,0.2){\tiny{$\supseteq A$}}
 \put(4,0){\circle{0}}
 
 \put(4.3,0){\small{$\ldots$}}
 
 \put(5,0){\circle{0}}
 \qbezier(5, 0)(5.5,-0.5)(6,  0)
 \put(5.2,0.0){\tiny{$\supseteq A$}}
 \put(6,0){\circle{0}}
 \put(6,0){\circle*{0.15}}
 \put(6,-0.4){\footnotesize{$l(\rho_{|n-1})$}}
 \qbezier(6, 0)(6.5,0.5)(7,  0)
 \put(6.2,0.3){\tiny{$\supseteq A$}}
 \put(7,0){\circle*{0.15}}
 \put(7.3,0){\small{$l(\rho)$}}
\end{picture}
\caption{The situation from \eqref{eq:condforin*} drawn in
  $\mmodel{M}$, i.e. the dots/circles belongs to $S$, and the links
  are links in $\rel{R}^D$}
\label{fig:possiblepaths}
\end{figure}
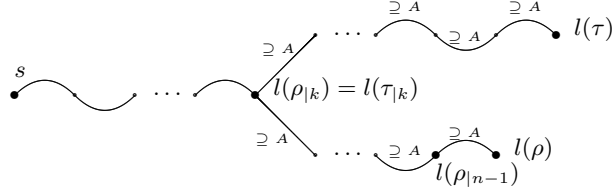%
As stated, we have to prove that
$\rel{R}^*{}^D_A = \biginter_{a \in A} \rel{R}^*{}^D_a$ for every $A
\in \powerne{\agents}$. The left-to-right inclusion immediately
follows from \textbf{(DC)}. For the converse, assume that $(\rho,
\tau) \in \rel{R}^*{}^D_a$ holds for every $a \in A$.  Then, for every
$a\in A$, according to \eqref{eq:condforin*}, there exists $k_a \geq
0$ such that $\rho_{|k_a} = \tau_{|k_a}$ and $\{a\}\subseteq
sl(\tau_{|i}), sl(\rho_{|j})$ for every $|\tau| \geq i > k_a$ and
every $|\rho| \geq j>k_a$. Now, let $k$ be the largest $k_a$
satisfying this condition (such a $k$ exists since
$\mmodel{M}^*$ is tree-like). Then, $\rho_{|k} = \tau_{|k}$, and for
every $a \in A$, the inclusions $\{a\}\subseteq sl(\tau_{|i})$ and
$\{a\}\subseteq sl(\rho_{|j})$ hold for every $|\tau| \geq i > k$ and
every $|\rho| \geq j >k$. Therefore, condition \eqref{eq:condforin*}
is fulfilled for $A$ and $k$, and hence $(\rho, \tau)\in
\rel{R}^*{}^D_A$, as desired.

 Finally, it remains to prove that $\mmodel{M}^*$ satisfies $\theta$.
 From \eqref{eq:condforin*} we see, that if $(\rho, \tau) \in
 \rel{R}^*{}^D_A$, then $(l(\rho), l(\tau)) \in \rel{R}^D_A$, since every
 $\rel{R}^D_A$ is an equivalence relation. It is now easy
 to check that the relation $Z = \crh{(\rho, l(\rho))}{\rho \in S^*}$
 is a bisimulation between $\mmodel{M}^*$ and $\mmodel{M}_s$.  Since
 $(s, l(s)) \in Z$, it follows that \sat{M^*}{s}{\theta}, and we are
 done.
\end{proof}

\begin{theorem}
 \label{thr:models_equal_Hintikka}
 Let $\theta$ be a \CMAELCD-formula.  Then, $\theta$ is satisfiable in
 a CMAEHS iff it is satisfiable in a CMAEM.
\end{theorem}

\begin{proof}
  Immediate from Lemmas~\ref{lem:models_to_Hintikka},
  \ref{lem:Hintikka_to_psmodels}, and \ref{lem:psmodels_to_models}.
\end{proof}

\section{Tableau procedure for testing satisfiability in \CMAELCD}
\label{sec:tableaux}

In this section, we present our tableau algorithm for checking
(constructive) satisfiability of formulae of \CMAELCD.  We start off
by explaining the general philosophy underlying our tableau procedure
and then present it in detail.

%%%%%%%%%%%%%%%%%%%%%%%%%%%%%%%%%%%%%%%%%%%
%%%%%%%%%%%%%%%%%%%%%%%%%%%%%%%%%%%%%%%%%%%
\subsection{Basic ideas and overview of the tableau procedure}
\label{subsec:overview_of_procedure}

Traditionally, the propositional tableau method works by decomposing
the formula whose satisfiability is being tested into its $\alpha$-,
resp. $\beta$- components -- repeatedly, until producing all full
expansions of that formula. All these components belong to the closure
of the input formula. When the closure is finite (as it is usually the
case with modal and temporal logics) the termination of the
tableau-building procedure is guaranteed because there are only
finitely many full expansions.

Furthermore, in the tableau method for the classical propositional
logic that decomposition into components produces a tree representing
an exhaustive search for a Hintikka set, the propositional analogue of
Hintikka structures, for the input formula. If at least one branch of
the tree remains open, it produces a full expansion of the input formula, which is a Hintikka set for this formula. In this case, the formula is
pronounced satisfiable; otherwise, it is declared unsatisfiable. In
the case of modal and temporal logics, local decomposition steps,
producing full expansions, are interleaved with steps along the
accessibility/transition relations, producing sets of formulae that
are supposed to be true at \emph{successors} of the current
state. These sets are subjected, again, to local decomposition into
components, eventually producing their full expansions, etc. In order
to distinguish fully expanded sets from those produced after
transition to successors, we will deal with two types of nodes of the
tableau, respectively called `states' and `prestates'. In order to
ensure termination of the construction process, we will systematically
reuse states and prestates labelled with the same sets of formulae.

The tableau procedure for testing a formula $\theta$ for
satisfiability attempts to construct a non-empty graph
$\tableau{T}^{\theta}$ (called itself a \fm{tableau}) representing
``sufficiently many'' CMAEHSs for $\theta$ in the sense that if
$\theta$ is satisfiable in any CMAEHS, then it is satisfiable in a
CMAEHS represented by the tableau. The procedure consists of three
major sub-procedures, or phases: \fm{construction}, \fm{prestate
  elimination}, and \fm{state elimination}.
During the construction phase, we build the \emph{pretableau}
$\tableau{P}^{\theta}$---a directed graph whose nodes are sets of
formulae of two types: \fm{states}\footnote{From now on we will use
  the term ``state'' in two related but distinct senses: as a state of a tableau
  and as a state of a semantic structure (frame, model, Hintikka
  structure). The use of term ``state'' will usually be clear from
  the context or explicitly specified.
  } 
  and \fm{prestates}, as
explained above. States represent (labels of) states of the CMAEHSs
that the tableau attempts to construct, while prestates are only used
temporarily, during the construction phase.

During the prestate elimination phase, we create a smaller graph
$\tableau{T}_0^{\theta}$ out of $\tableau{P}^{\theta}$, called the
\fm{initial tableau for $\theta$}, by eliminating all the prestates of
$\tableau{P}^{\theta}$ and adjusting its edges, as prestates have
already fulfilled their role of keeping the graph finite and can,
therefore, be discharged.

In the case of classical propositional logic, the only reason why it
may turn out to be impossible to produce a Hintikka set for the input
formula is that every attempt to build such a set results in a
collection of formulae containing a patent inconsistency.  In the case
of logics with fixpoint-definable operators, such as \CMAELCD, there
are two other reasons for a tableau not to correspond to any Hintikka
structure for the input formula.  The first one has to do with
\fm{realization of eventualities} ---formulas of the form $\neg
\commonk{A} \vp$, whose truth condition requires that $\neg \vp$
``eventually'' becomes true --- in the tableau graph.  Applying
decomposition rules to eventualities in the construction of the
tableau can postpone indefinitely the realization by keeping
``promising'' that the realization will happen further down the line,
while that ``promise'' never becomes fulfilled.  Therefore, a ``good''
tableau should not contain states with unrealized eventualities.  The
other additional reason for the resultant tableau not to represent a
Hintikka structure is that some states do not have all the successors
they would be required to have in a corresponding Hintikka structure
(for example, because those successors have been removed for not
realizing eventualities).

During the state elimination phase, we remove from
$\tableau{T}_0^{\theta}$ all states, if any, that cannot be satisfied
in any CMAEHS for any of the reasons suggested above and discussed in
more detail further (excluding patently inconsistent sets, which are
removed ``on the fly'' during the construction phase).  The
elimination procedure results in a (possibly empty) subgraph
$\tableau{T}^{\theta}$ of $\tableau{T}_0^{\theta}$, called the
\de{final tableau for $\theta$}. If some state $\De$ of
$\tableau{T}^{\theta}$ contains $\theta$, it is declared satisfiable;
otherwise, $\theta$ is declared unsatisfiable.

The logic \CMAELCD\ involves modal operators over equivalence
relations, and thus invokes some typical complications in the
tableau-building procedures associated with inverse-looking
modalities, see e.g. \cite{Gore98}: every box occurring in the label
of a descendant state has a backwards effect on all predecessor
states, incl. the current state. In order to deal with these
complications we must either organize a mechanism for backtracking and
\emph{backwards propagation of box-formulae}, or a mechanism for
anticipation of the occurrence of such boxes in the future, coming
from subformulae of formulae in the label of the current state, based
on \emph{analytic cut rules}. We will adopt here the latter approach, which is easier to describe and implement into what we call a \emph{diamond-pro\-pa\-ga\-ting} procedure, by employing suitably restricted analytic cut rules to maintain the efficiency of the procedure, but later we will briefly discuss the former
alternative, too. The two procedures only differ in the construction phase;
the prestate and state elimination phases are common to both. The need and use of analytic cut rules is illustrated later in Example \ref{ex:cut}.

%%%%%%%%%%%%%%%%%%%%%%%%%%%%%%%%%%%%%%%%%%%%%%
%%%%%%%%%%%%%%%%%%%%%%%%%%%%%%%%%%%%%%%%%%%%%%
\subsection{Cut-saturated sets and expansions}
\label{subsec:cut_saturated}

The application of the analytic cut, mentioned above, is implemented by imposing an additional \emph{cut-saturating rule} on the construction of the full expansions of a given set of formulae. In order to prevent the unnecessary swelling and proliferation of states, we will restrict the application of  that rule by imposing generic restrictions which, on the other hand, should be sufficiently relaxed to guarantee the completeness of the tableau procedure. These generic conditions, which will be specified later, will be imposed separately on the two types of box-formulae in \CMAELCD, viz. $\distrib{A}$-formulae and on $\commonk{A}$-formulae. 

\begin{definition}
  \label{def:cmaelcd_fully_expanded} Given restrictive conditions $C_1$ and $C_2$,  a set $\De$ of \CMAELCD-formulas is \de{$(C_1,C_2)$-\funny} if it
  satisfies the following conditions, where 
  \subf{\psi} is the set of subformulae of a formula $\psi$:  
  
\begin{enumerate}[label=CS\arabic*, start=0]

\item \label{it:FE1} $\De$ is fully expanded (recall
  Definition~\ref{def:cmaelcd_downward_saturated}).
  
\item \label{it:FE2} For any $\distrib{A}\vp\in \subf{\psi}$ where
  $\psi\in\De$, if condition $C_1$ holds then either $\distrib{A} \vp
  \in \De$ or $\neg \distrib{A} \vp \in \De$.
  
\item \label{it:FE3} For any $\commonk{A}\vp\in \subf{\psi}$ where
  $\psi\in\De$, if condition $C_2$ holds then either $\commonk{A} \vp
  \in \De$ or $\neg \commonk{A} \vp \in \De$.
 \end{enumerate}
\end{definition}

We note that \ref{it:FE2} and \ref{it:FE3} are semantically sound
rules, no matter what $C_1$ and $C_2$ are, as they cannot make a
tableau closed if the input formula is satisfiable. On the other hand,
if $C_1$ and $C_2$ are too strong, that may prevent the tableau from closing and thus yield an incomplete
tableau procedure, as will become apparent later.  Again, the reason
we would want to make $C_1$ and $C_2$ as strong as possible is to
avoid branching on too many formulae, causing an unnecessary large
state space and resulting in a practically less efficient procedure.

Hereafter, we will omit the explicit mention of the conditions $C_1$
and $C_2$, unless necessary. In fact, for now we can assume both $C_1$
and $C_2$ to be \textsf{True}, but later we will introduce non-trivial
restrictive conditions.

%%%%%%%%%%%%%%%%%%%%%%%%%%%%%%
\begin{definition}\label{def:full_expansion}
The family $\cse{\Ga}$ of \de{\funny expansions (CS-expansions) of a set of formulae $\Ga$} is defined by expanding the procedure $\fexp$ with the following two set-replacement rules, again applied to a non-deterministically chosen set $\Phi$ from the current family and a formula $\psi \in \Phi$:  
  \begin{enumerate}
\item For any formula $\di{A}\vp$ that is a subformula of $\psi$ such that $C_1$ is satisfied, replace $\Phi$ with the two extensions of $\Phi$ obtained by adding respectively
  $\di{A}\vp$ and $\neg\di{A}\vp$ to it.
\item For any formula $\co{A}\vp$ that is a subformula of $\psi$ such that $C_2$ is satisfied, replace $\Phi$ with the two extensions of $\Phi$ obtained by adding respectively
  $\co{A}\vp$ and $\neg\co{A}\vp$ to it.
\end{enumerate}
\end{definition}

It is clear from the definition that all sets in $\cse{\Ga}$ are 
$(C_1,C_2)$-\funny.

\begin{definition}
 \label{def:extended_closure}

 The \de{extended closure} of $\theta$, denoted $\ecl{\theta}$, is the
 smallest set such that $\vp, \neg \vp \in \ecl{\theta}$ for every $\vp
 \in \cl{\theta}$. The extended closure $\ecl{\Ga}$ of a set of
 formulae $\Ga$ is defined likewise.
\end{definition}

The following is immediate from the definitions. 

\begin{lemma}
Every CS-expansion of a set of formulae $\Ga$ is a subset of $\ecl{\Ga}$.
\end{lemma}

\begin{lemma}
  \label{lem:size_of_closure}
  For any \CMAELCD-formula $\theta$, the size of (i.e., number of
  formulae in) the extended closure of $\theta$ is \bigo{k \cdot
    \card{\theta}}, where $k$ is the number of agents occurring in
  $\theta$.
\end{lemma}

\begin{proof}
  Straightforward.
\end{proof}

%%%%%%%%%%%%%%%%%%%%%%%%%%%%%%%%%%%
\subsubsection{Construction phase}
\label{subsec:diam_proparating}

As already mentioned, a tableau algorithm attempts to produce a
compact representation of ``sufficiently many'' CMAEHSs for the input
formula; in this attempt, it sets in motion an exhaustive search for
such CMAEHSs. As a result, the pretableau $\tableau{P}^{\theta}$ built
at this phase contains two types of edge, as well as two types of node
(states and prestates; see above).

One type of edge, depicted by unmarked, dashed uni-directed arrows
$\dashrightarrow$, represents the search dimension of the
tableaux.  The exhaustive search considers all possible alternatives
arising when prestates are expanded into states by branching in the
``disjunctive'' cases. Thus, when we draw unmarked arrows from a
prestate \Ga\ to each state from a set of states $X$, this intuitively
means that, in any CMAEHS, a state satisfying \Ga\ has to satisfy at
least one of the states in $X$.

The second type of edge represents transition relations in the CMAEHSs
that the procedure attempts to build.  Accordingly, this type of edges is
represented by solid, uni-directed arrows, $\longrightarrow$, 
marked with formulae whose
presence in one of the end nodes requires the presence in the tableau
of the other end node, reachable by a particular relation.
Intuitively, if $\neg \distrib{A} \vp \in \De$ for some state $\De$,
then some (state obtained from a) prestate $\Ga$ containing $\neg \vp$
must be accessible from $\De$ by relation $\rel{R}^D_A$.  We
mark these arrows with the respective formulae $\neg
\distrib{A} \vp$ in order to keep track of the specific reason for creating 
 that particular state. That information will be needed during the
elimination phases.

We now turn to presenting the rules of the ``diamond-propagating''
construction phase, each of which creates a different type of edge, as
discussed above.  The first rule, \RRule{SR}, prescribes how to create
states from prestates,  while \RRule{DR} expands prestates into states.

\medskip

\textbf{Rule} \RRule{SR} Given a prestate $\Ga$, such that \RRule{SR} has
not been applied to it before, do the following:
\begin{enumerate}
\item Add to the pretableau all CS-expansions $\De$ of $\Ga$; declare these to be
  \de{states};
\item For each so obtained state $\De$, put $\Ga\dashrightarrow \De$;
\item If, however, the pretableau already contains a state $\De' = \De$,
  then do not create a new state, but put $\Ga \dashrightarrow
  \De'$.
\end{enumerate}

We denote by $\st{\Ga}$ the (finite) set \crh{\De}{\Ga \dashrightarrow \De}.

\medskip

\textbf{Rule} \RRule{DR}: Given a state $\De$ such that $\neg
\distrib{A} \vp \in \De$ and \RRule{DR} has not been applied to $\De$
with respect to $\neg \distrib{A} \vp$ before, do the following:

\begin{enumerate}
\item Add to the pretableau the set 
        $\Ga = \set{\neg \vp} \union\crh{\distrib{A'} \psi\in \De}{A'\subseteq A} \union 
  \crh{\neg \distrib{A'}\psi\in \De}{A'\subseteq A \text{ and }\neg \distrib{A'} \psi \ne \neg \distrib{A} \vp  
}
    \union 
  \crh{\neg \commonk{A'}\psi\in \De}{A'\cap A \neq \emptyset}$
  and declare this set to be a \emph{prestate}. 
\item Put $\DtoD{\De}{\Ga}{A}{\vp}$.
\item If, however, the pretableau already contains a prestate $\Ga' =
  \Ga$, then do not create a new prestate, but put $\DtoD{\De}{\Ga'}{A}{\vp}$.
\end{enumerate}

When building a tableau for a formula $\theta$, the construction phase
begins with creating a single prestate \set{\theta}.  Afterwards, we
alternate between \RRule{SR} and \RRule{DR}: first, \RRule{SR} is
applied to the prestates created at the previous stage of the
construction, then \RRule{DR} is applied to the states created at the
previous stage.

The construction phase is completed when every prestate required
to be added to the pretableau has already been added (as prescribed in
item 3 of \RRule{SR}) and \RRule{DR} does not apply to any of the states
with respect to any of the formulae.

\medskip

\begin{example}\label{ex1}
Let us construct the pretableau for the
formula $\theta = \neg \distrib{\set{a,c}} \commonk{\set{a,b}} p \land
\commonk{\set{a,b}}(p \con q)$, assuming that $\agents = \set{a, b,
  c}$.  To save space, we replace $\theta$ by the set of its conjuncts
$\Theta = \set{\neg \distrib{\set{a,c}} \commonk{\set{a,b}} p,
  \commonk{\set{a,b}}(p \con q)}$. 

Here and further on in the examples, we let $\underline{\commonk{A}\vp}$ denote the set $\{\commonk{A}\vp, \vp\}\cup\bigcup_{a\in A}\distrib{a}\commonk{A}\vp$.
Figure~\ref{fig:ex1} shows the pretableau for $\Theta$.

\begin{figure}[h!]
\begin{minipage}[c]{0.45\linewidth}
\begin{equation*}
\tiny{
\xymatrix{
&\Ga_1\ar@{-->}[dl]\ar@{-->}[d]\ar@{-->}[dr]\\
\De_2\ar^{\chi}[d]\ar@/_2.0pc/[dddd]_{\chi_a}&\De_3\ar[ddddl]_(0.3){\chi}\ar[ddddr]^(0.3){\chi_b}&\De_1\ar^{\chi}[d]&\\   
\Ga_3\ar@{-->}[d]
 &&\Ga_2 \\
\De_4\ar@/_0.5pc/[dd]_(0.3){\chi_a}&
&\De_5\ar@/^0.5pc/[dd]^{\chi_b}\\
&&&\\
\Ga_5\ar@{-->}[uu]\ar@{-->}[uurr]&&\Ga_6\ar@{-->}[uull]\ar@{-->}[uu]\\ 
}
}
\end{equation*}
\end{minipage}
\hfill
\begin{minipage}[c]{0.45\linewidth}
\tiny{
\begin{align*}
\chi &= \neg\distrib{\{a,c\}}\commonk{\{a,b\}}p\\\
\chi_a &=\neg\distrib{a}\commonk{\{a,b\}}p\\
\chi_b &= \neg\distrib{b}\commonk{\{a,b\}}p\\
\\
\Ga_1&=\{\neg\distrib{\{a,c\}}\commonk{\{a,b\}}p,\commonk{\{a,b\}}(p\land q)\}\\
\De_1&=\{\neg\distrib{\{a,c\}}\commonk{\{a,b\}}p,\underline{\commonk{\{a,b\}}(p\land q)}, p, q, \underline{\commonk{\{a,b\}}p}\}\\
\De_2&=\{\neg\distrib{\{a,c\}}\commonk{\{a,b\}}p,\underline{\commonk{\{a,b\}}(p\land q)}, p, q, \neg\commonk{\{a,b\}}p, \neg\distrib{a}\commonk{\{a,b\}}p\}\\
\De_3&=\{\neg\distrib{\{a,c\}}\commonk{\{a,b\}}p,\underline{\commonk{\{a,b\}}(p\land q)}, p, q, \neg\commonk{\{a,b\}}p, \neg\distrib{b}\commonk{\{a,b\}}p\}\\
\Ga_2&=\{ \neg\commonk{\{a,b\}}p, \distrib{a}\commonk{\{a,b\}}(p\land q), \distrib{a}\commonk{\{a,b\}}p \}\\
\Ga_3&=\{ \neg\commonk{\{a,b\}}p, \distrib{a}\commonk{\{a,b\}}(p\land q), \neg\distrib{a}\commonk{\{a,b\}}p \}\\
\De_4&=\{\neg\commonk{\{a,b\}}p,\underline{\commonk{\{a,b\}}(p\land q)}, p, q, \neg\distrib{a}\commonk{\{a,b\}}p\}\\
\De_5&=\{\neg\commonk{\{a,b\}}p,\underline{\commonk{\{a,b\}}(p\land q)}, p, q, \neg\distrib{b}\commonk{\{a,b\}}p\}\\
\Ga_5&=\{ \neg\commonk{\{a,b\}}p, \distrib{a}\commonk{\{a,b\}}(p\land q) \}\\
\Ga_6&=\{ \neg\commonk{\{a,b\}}p, \distrib{b}\commonk{\{a,b\}}(p\land q) \}\\
\end{align*}}
\end{minipage}
\caption{The pretableau for $\set{\neg \distrib{\set{a,c}} \commonk{\set{a,b}} p, \commonk{\set{a,b}}(p \con q)}$}
\label{fig:ex1}
\end{figure}
\end{example}
\subsubsection{Prestate elimination phase}

At this phase, we remove from pretableau $\tableau{P}^{\theta}$ all
the prestates and unmarked arrows, by applying the following rule (the
resultant graph is denoted $\tableau{T}_0^{\theta}$ and is called the
\emph{initial tableau}):

\smallskip

\RRule{PR} For every prestate $\Ga$ in $\tableau{P}^{\theta}$, do the
following:

\begin{enumerate}
\item Remove $\Ga$ from $\tableau{P}^{\theta}$;
\item If there is a state $\De$ in $\tableau{P}^{\theta}$ with $\De
 \stackrel{\chi}{\longrightarrow} \Ga$, then for every state $\De' \in
 \st{\Ga}$, put $\De \stackrel{\chi}{\longrightarrow} \De'$;
\end{enumerate}

\begin{example}\label{ex1-cont}
We continue Example~\ref{ex1} by creating the initial tableau for $\Theta = \set{\neg \distrib{\set{a,c}} \commonk{\set{a,b}} p,$ $\commonk{\set{a,b}}(p \con q)}$ out of the pretableau in Figure~\ref{fig:ex1}. Again we let $\underline{\co{A}\vp}$ denote the set consisting of $\co{A}\vp$ and its $\alpha$-components. Figure~\ref{fig:ex1-cont} shows the resulting initial tableau.
\begin{figure}[!h]
\begin{minipage}[c]{0.45\linewidth}
\begin{equation*}
\tiny{
\xymatrix{
\De_2\ar[dd]_(0.5){\chi,\chi_a}\ar[ddrr]^(0.2){\chi_a}
&\De_3\ar[ddl]^(0.6){\chi, \chi_b}\ar[ddr]^(0.3){\chi,\chi_b}&\De_1\ar^{\chi}[d]\\
&&\\
\De_4\ar@(dl,dr)[]_{\chi_a}\ar@/^/^{\chi_a}[rr]&
&\De_5\ar@(dl,dr)[]_{\chi_b}\ar@/^/^{\chi_b}[ll]
}
}
\end{equation*}
\end{minipage}
\hfill
\begin{minipage}[c]{0.45\linewidth}
\tiny{
\begin{align*}
\chi &= \neg\distrib{\{a,c\}}\commonk{\{a,b\}}p\\\
\chi_a &=\neg\distrib{a}\commonk{\{a,b\}}p\\
\chi_b &= \neg\distrib{b}\commonk{\{a,b\}}p\\
\\
\De_1&=\{\neg\distrib{\{a,c\}}\commonk{\{a,b\}}p,\underline{\commonk{\{a,b\}}(p\land q)}, p, q, \underline{\commonk{\{a,b\}}p}\}\\
\De_2&=\{\neg\distrib{\{a,c\}}\commonk{\{a,b\}}p,\underline{\commonk{\{a,b\}}(p\land q)}, p, q, \neg\commonk{\{a,b\}}p, \neg\distrib{a}\commonk{\{a,b\}}p\}\\
\De_3&=\{\neg\distrib{\{a,c\}}\commonk{\{a,b\}}p,\underline{\commonk{\{a,b\}}(p\land q)}, p, q, \neg\commonk{\{a,b\}}p, \neg\distrib{b}\commonk{\{a,b\}}p\}\\
\De_4&=\{\neg\commonk{\{a,b\}}p,\underline{\commonk{\{a,b\}}(p\land q)}, p, q, \neg\distrib{a}\commonk{\{a,b\}}p\}\\
\De_5&=\{\neg\commonk{\{a,b\}}p,\underline{\commonk{\{a,b\}}(p\land q)}, p, q, \neg\distrib{b}\commonk{\{a,b\}}p\}\\
\end{align*}}
\end{minipage}
\caption{The initial tableau for $\set{\neg \distrib{\set{a,c}} \commonk{\set{a,b}} p, \commonk{\set{a,b}}(p \con q)}$}
\label{fig:ex1-cont}
\end{figure}
\end{example}

\subsubsection{State elimination phase}

During this phase, we remove from $\tableau{T}_0^{\theta}$ states that are not satisfiable in any CMAEHS. Of course, when a state is removed, so are all of its incoming and outgoing arrows.
 
There are two reasons why a state
$\De$ of $\tableau{T}_0^{\theta}$ might turn out to be unsatisfiable:
either because $\De$ needs, in order to satisfy some diamond-formula, a successor state that has been eliminated, or because $\De$ contains an eventuality
that is not realized in the tableau. Accordingly, we have two
elimination rules \RRule{E1} and \RRule{E2}.

Formally, the state elimination phase is divided into stages; we start
at stage 0 with $\tableau{T}_0^{\theta}$; at stage $n+1$, we remove
from the tableau $\tableau{T}_n^{\theta}$ obtained at the previous
stage exactly one state, by applying one of the elimination rules,
thus obtaining the tableau $\tableau{T}_{n+1}^{\theta}$. We state the
rules below, where $S_m^{\theta}$ denotes the set of states of
$\tableau{T}_{m}^{\theta}$.

\smallskip

\RRule{E1} If $\De\in S^{\theta}_n$ contains a formula $\chi = \neg
\distrib{A} \vp$ such that there is no $\De
\stackrel{\chi}{\longrightarrow} \De'$, where $\De' \in S^{\theta}_n$,
then obtain $\tableau{T}_{n+1}^{\theta}$ by eliminating $\De$ from
$\tableau{T}_n^{\theta}$.

\smallskip

For the other elimination rule, we need the concept of eventuality
realization.

\begin{definition}
  \label{def:realization}
  The eventuality $\xi= \neg \commonk{A} \vp$ \de{is realized at $\De$ in
    $\tableau{T}^{\theta}_n$} if either $\neg \vp \in \De$ or there
  exists in $\tableau{T}^{\theta}_n$ a finite number of states $\De_0, \De_1,
  \ldots, \De_m$ such that 
  $\De_0 = \De$; 
  $\neg \vp \in \De_m$; and, for every $0 \leq i < m$, 
  $\xi \in \De_i$ 
 and  
   there exists $\chi_i = \neg \distrib{a_i}
  \psi_i$ such that $a_i \in A$ and $\De_i
  \stackrel{\chi_i}{\longrightarrow} \De_{i+1}$.
\end{definition}

We can now state the rule.

\medskip

\RRule{E2} If $\De \in S_n^{\theta}$ contains an eventuality $\neg
\commonk{A} \vp$ that is not realized at $\De$ in
$\tableau{T}_n^{\theta}$, then obtain $\tableau{T}_{n+1}^{\theta}$ by
removing $\De$ from $\tableau{T}_n^{\theta}$.

\medskip

We check for realization of $\neg \commonk{A} \vp$ by running the
following, \emph{global} procedure that marks all states of
$\tableau{T}_n^{\theta}$ realizing $\neg \commonk{A} \vp$ in
$\tableau{T}_n^{\theta}$.
Initially, we mark all $\De \in S_n^{\theta}$ such that $\neg \vp \in \De$. Then,
we repeatedly do the following: if $\De \in S_n^{\theta}$ contains $\neg \commonk{A} \vp$ and is unmarked yet, 
but there exists at least one $\De'$ such that $\De
\stackrel{\neg \distrib{a} \psi}{\longrightarrow} \De'$, for some formula $\psi$ and $a\in A$  and $\De'$ is marked, we
mark $\De$. The procedure is over when no more states get marked. Note
that marking is carried out with respect to a fixed eventuality $\xi$
and is, therefore, repeated each time we want to check realization of
an eventuality (see reasons further).

We have so far described elimination rules; to describe the state
elimination phase as a whole, we need to specify the order of their
application.  We have to be careful since, having applied \RRule{E2},
we could have removed all the states accessible from some $\De$ along
the arrows marked with some formula $\chi$; hence, we need to reapply
\RRule{E1} to the resultant tableau to remove such $\De$'s. Conversely,
after having applied \RRule{E1}, we could have thrown away some states
that were needed for realizing certain eventualities; hence, we need
to reapply \RRule{E2}. Moreover, we cannot terminate the procedure
unless we have checked that \emph{all} eventualities are realized.
Therefore, we apply \RRule{E1} and \RRule{E2} in a dovetailed sequence
that cycles through all the eventualities. More precisely, we arrange
all eventualities occurring in the states of $\tableau{T}_0^{\theta}$
in a list $\xi_1, \ldots, \xi_m$.  Then, we proceed in cycles. Each
cycle consists of alternatingly applying \RRule{E2} to the pending
eventuality (starting with $\xi_1$), and then applying \RRule{E1} to
the resulting tableau, until all the eventualities have been dealt
with.  These cycles are repeated until no state is removed throughout
a whole cycle. When that happens, the state elimination phase is over.

The graph produced at the end of the state elimination phase is called
the \fm{final tableau for $\theta$}, denoted by
$\tableau{T}^{\theta}$, and its set of states is denoted by
$S^{\theta}$.

\begin{definition}
 The final tableau $\tableau{T}^{\theta}$ is \de{open} if $\theta \in
 \De$ for some $\De \in S^{\theta}$; otherwise, $\tableau{T}^{\theta}$
 is \de{closed}.
\end{definition}

The tableau procedure returns ``no'' (not satisfiable) if the final tableau is closed;
otherwise, it returns ``yes'' (satisfiable) and, moreover, provides sufficient
information for producing a finite model satisfying $\theta$; that
construction is sketched in Section \ref{subsec:completeness}.

\begin{example}
  We will continue to make the final tableau for the formulae $\Theta$
  considered in Example~\ref{ex1} and Example~\ref{ex1-cont}. The
  state elimination procedure starts with the initial tableau given in
  Figure~\ref{fig:ex1-cont}. During the state-elimination phase, state
  $\De_1$ gets removed due to \RRule{E1}, since it does not have any
  successor states along an arrow labelled with $\chi$, while states
  $\De_2, \De_3, \De_4$ and $\De_5$ are eliminated due to \RRule{E2},
  as all of them contain the unrealized eventuality $\neg
  \commonk{\set{a,b}} p$. Thus, the final tableau for $\Theta$ is an
  empty graph; therefore, $\Theta$ is \emph{unsatisfiable}.
\end{example}

%%%%%%%%%%%%%%%%%%%%%%%%%%%%%%%%%%%%%%%%%%%%
%%%%%%%%%%%%%%%%%%%%%%%%%%%%%%%%%%%%%%%%%%%%
\section{Soundness and completeness of the tableau}
\label{sec:scc}

\subsection{Soundness}
\label{subsec:soundness}

Technically, soundness of a tableau procedure amounts to claiming that
if the input formula $\theta$ is satisfiable, then the final tableau
$\tableau{T}^{\theta}$ is open.

Before going into the technical details, we give an informal outline
of the proof.  The tableau procedure for the input formula $\theta$
starts off with creating a single prestate \set{\theta}.  Then, we
expand \set{\theta} into states, each of which contains $\theta$.  To
establish soundness, it suffices to show that at least one of these
states survives to the end of the procedure and is, thus, part of the
final tableau. 

We start out by showing (Lemma~\ref{lem:expansion}) that if a prestate
$\Ga$ is satisfiable, then at least one state created from $\Ga$ using
\RRule{SR} is also satisfiable.  In particular, this ensures that if
$\theta$ is satisfiable, then so is at least one state obtained by
\RRule{SR} from $\set{\theta}$.  To ensure soundness, it suffices to
prove that this state never gets eliminated from the tableau.

To that end, we first show (Lemma~\ref{lem:DR_sound}) that, given a
satisfiable state $\De$, all the prestates created from $\De$ in
accordance with \RRule{DR}---each prestate being associated with a
formula of the form $\neg \distrib{A} \vp$---are satisfiable;
according to Lemma~\ref{lem:expansion}, each of these prestates will
give rise to at least one satisfiable state.  It follows that, if a
tableau state $\De$ is satisfiable, then every successor of $\De$ in
the initial tableau will have at least one satisfiable successor
reachable by an arrow associated with each formula of the form $\neg
\distrib{A} \vp$ belonging to $\De$. Hence, if $\De$ is satisfiable, it
will not be eliminated on account of \RRule{E1}.

Second, we show that no satisfiable states contain unrealized
eventualities (in the sense of Definition~\ref{def:realization}), and
thus cannot be removed from the tableau on account of \RRule{E2}.
Thus, we show that a satisfiable state of the pretableau
(equivalently, initial tableau) cannot be removed on account of any of
the state elimination rules and, therefore, survives to the end of the
procedure. In particular, this means that at least one state obtained
from the initial prestate $\theta$, and thus containing $\theta$,
survives to the end of the procedure. Hence, the final tableau for
$\theta$ is open, as desired.

We emphasize again that the claims mentioned above, and their proofs,
do not depend on the application (or not) of the cut rules
\ref{it:FE2} and \ref{it:FE3}, because they are sound, since
$\gamma\lor\neg\gamma$ is valid for any formula
$\gamma$. Therefore, these results are unaffected by the restrictive
conditions $C_1$ and $C_2$ for their application.

We now proceed with the technical details. 

\begin{lemma}
 \label{lem:expansion}
 Let $\Ga$ be a prestate of $\tableau{P}^{\theta}$ such that
 \sat{M}{s}{\Ga} for some CMAEM \mmodel{M} and $s \in
 \mmodel{M}$.  Then:

 \begin{enumerate}
\item   \sat{M}{s}{\De} holds for at least one $\De \in \st{\Ga}$.
\item Moreover, if $\neg\co{A}\vp\in\Ga$ and $\sat{M}{s}{\neg\vp}$,
  then $\De$ can be chosen so that $\neg\vp\in\De$.
\item If $\neg\co{A}\vp\in\Ga$ while none of $\neg\co{A}\vp$'s
  $\beta$-components are in $\Ga$, then for every $a\in A$, if
  $\sat{M}{s}{\neg\di{a}\co{A}\vp}$ then $\De$ can be chosen so that
  either $\neg\di{a}\co{A}\vp\in\De$ or $\neg\vp\in\De$.
\end{enumerate}

\end{lemma}
\begin{proof}
  Straightforward from the definition of $\cse{\Ga}$ and using
  Proposition~\ref{lem:full_expansion}.
\end{proof}

\begin{lemma}
 \label{lem:DR_sound}
 Let $\De \in S_0^{\theta}$ be such that \sat{M}{s}{\De} for some
 CMAEM \mmodel{M} and $s \in \mmodel{M}$, and let $\neg \distrib{A}
 \vp \in \De$.  Then, there exists $t \in \mmodel{M}$ such that $(s,
 t) \in \rel{R}^D_A$ and \sat{M}{t}{\Ga}, for a set $\Ga$ defined
 according to the rule \RRule{DR} applied to $ \De$ and $\neg
 \distrib{A} \vp$:
 
 $\Ga =  \set{\neg \vp} \union \crh{\distrib{A'} \psi\in \De}{A'\subseteq A}
   \union \crh{\neg \distrib{A'}
     \psi\in\De}{A'\subseteq A\text{ and } \neg \distrib{A'} \psi \ne \neg \distrib{A} \vp}
         \union \crh{\neg \commonk{A'}\psi\in \De}{A'\cap A \neq \emptyset}$
\end{lemma}

\begin{proof}
  Easily follows from the semantics of the epistemic operators and the definition of CMAEM.
\end{proof}

\begin{lemma}\label{lemma1}
  Let $\De\in\tableau{S}_0^{\theta}$, let $\neg\co{A}\vp,
  \neg\di{a}\co{A}\vp \in\De$, and let, furthermore,
  $\DtoD{\De}{\Ga}{a}{\co{A}\vp}$ for some prestate
  $\Ga\in\tableau{P}^{\theta}$. Assume that $\sat{M}{s}{\De}$ and
  $(s,s')\in\rel{R}^D_a$, for some model $\mmodel{M}$ and a pair of
  states $s,s'\in \mmodel{M}$; then $\sat{M}{s'}{\Ga}$.
\end{lemma}

\begin{proof} 

  Recall from the rule \RRule{DR} that $\Ga= \{\neg\co{A}\vp\} \cup
  \crh{\di{a}\gamma}{\di{a}\gamma\in\De}\cup
  \crh{\neg\di{a}\gamma}{\neg\di{a}\gamma\in\De, \neg\di{a}\gamma\neq
    \neg\di{a}\co{A}\vp} \cup
  \crh{\neg\co{a}\gamma}{\neg\co{a}\gamma\in \De}$.  The claim follows
  easily, because $\rel{R}^D_ a$ is an equivalence relation. Indeed,
  $\sat{M}{s'}{\neg\co{A}\vp}$ because every $A$-reachable state from
  $s$ is $A$-reachable from $s'$, too. Moreover, $(s',s'')\in
  \rel{R}^D_a$ iff $(s,s'')\in \rel{R}^D_a$, for all $s''$. Therefore,
  $\sat{M}{s'}{\chi}$ for all $\chi \in \Ga \setminus \{\neg
  \co{A}\vp\}$.
\end{proof}

\begin{lemma}
 \label{lm:E3_sound}
 Let $\De \in \tableau{S}_0^{\theta}$ be such that \sat{M}{s}{\De} for
 some CMAEM \mmodel{M} and $s \in \mmodel{M}$, and let $\neg
 \commonk{A} \vp \in \De$. Then there is a finite path in
 $\tableau{S}_0^{\theta}$ of satisfiable states that realizes
 $\neg\co{A}\vp$ at $\De$.
\end{lemma}

\begin{proof}
We start by proving the following:

\begin{quote}
  Let $\neg\co{A}\vp\in\Ga_1$ for some prestate
  $\Ga_1\in\tableau{P}^{\theta}$ such that $\Ga_1$ does not contain
  any of the $\beta$-components of $\neg\co{A}\vp$.  Suppose that
  $\sat{M}{s_1}{\Ga_1}$, and let $
  s_1\overset{{a_1}}{\longrightarrow}s_2\overset{{a_2}}{\longrightarrow}\ldots\overset{{a_{n-1}}}{\longrightarrow}s_n
  $ be a \emph{shortest} path in $\mmodel{M}$ that satisfies
  $\neg\co{A}\vp$, i.e., $\sat{M}{s_n}{\neg\vp}$, and for all $i< n$,
  the following hold: $\sat{M}{s_i}{\set{\neg\co{A}\vp, \vp}}$, and
  $(s_i,s_{i+1})\in \rel{R}^D_{a_i}$, for some $a_i\in A$. Then there
  exists a path 
\[
\De_1\overset{\neg\di{a_1}\co{A}\vp}{\longrightarrow}\De_2\overset{\neg\di{a_2}\co{A}\vp}{\longrightarrow}\ldots\overset{\neg\di{a_{n'-1}}\co{A}\vp}{\longrightarrow}\De_{n'},
\]
of satisfiable states in $\tableau{S}_0^{\theta}$, where $n'\leq n$, 
$\De_1\in\st{\Ga_1}$ and $\neg\vp\in \De_{n'}$.
\end{quote}

We prove the above claim by induction on $n$.

If $n = 1$, then $\sat{M}{s_1}{\neg \vp}$. Since $\neg\co{A}\vp\in
\Ga_1$ and $\sat{M}{s_1}{\Ga_1}$, Lemma \ref{lem:expansion} implies
that there is a $\De_1\in \st{\Ga_1}$ such that $\sat{M}{s_1}{\De_1}$
and $\neg \vp\in \De_1$. Thus $\De_1$ is the needed path in
$\tableau{S}_0^{\theta}$ that satisfies the claim above.

Assume now the claim holds for all $m < n$. Let
$\neg\co{A}\vp\in\Ga_1$, let $\sat{M}{s_1}{\Ga_1}$, and assume that none of
$\neg\co{A}\vp$'s $\beta$-components are in $\Ga_1$. Let the path in
$\mmodel{M}$ satisfying the eventuality $\neg\co{A}\vp$ be
$s_1\overset{{a_1}}{\longrightarrow}s_2\overset{{a_2}}{\longrightarrow}\ldots\overset{{a_{n-1}}}{\longrightarrow}s_n$,
where $n>1$.

Since $\sat{M}{s_1}{\set{\neg\di{a_1}\co{A}\vp, \neg\co{A}\vp}}$,
Lemma \ref{lem:expansion} implies the existence of
$\De_1\in\st{\Ga_1}$ in $\tableau{S}_0^{\theta}$ with
$\sat{M}{s_1}{\De_1}$, and $\neg\di{a_1}\co{A}\vp\in\De_1$ or $\neg\vp\in\De_1$. 
In the latter case, $\De_1$ is the needed path. In the former case, due
to Lemma \ref{lem:DR_sound}, there exists a prestate
$\Ga_2\in\tableau{T}^{\theta}$, with
$\DtoD{\De_1}{\Ga_2}{{a_1}}{\co{A}\vp}$; then,
$\neg\co{A}\vp\in\Ga_2$.  Note that $\Ga_2$ cannot contain any of
$\neg\co{A}\vp$'s $\beta$-components, since $\De_1$ contains
$\neg\di{a_1}\co{A}\vp$, and thus, it can contain at most one other
$\beta$-component, namely $\neg\vp$. But in that case we would have
that $\sat{M}{s_1}{\neg\vp}$, which contradicts the assumption. Lemma
\ref{lemma1} gives us $\sat{M}{s_2}{\Ga_2}$.

Thus, since
$s_2\overset{{a_2}}{\longrightarrow}\ldots\overset{{a_{n-1}}}{\longrightarrow}s_n$
is a path of length $n-1$ that realizes $\neg\co{A}\vp$ at $s_2$, the
induction hypothesis claims that there is a path of satisfiable states
in $\tableau{S}_0^{\theta}$,
\[
\De_2\overset{\neg\di{a_2}\co{A}\vp}{\longrightarrow}\ldots\overset{\neg\di{a_{n'-1}}\co{A}\vp}{\longrightarrow}\De_{n'},
\]
where $n'\leq n-1$, $\De_2\in\st{\Ga_2}$, $\neg\vp\in \De_{n'}$. 

Since $\Ga_1\dashrightarrow
\DtoD{\De_1}{\Ga_2}{a_1}{\co{A}\vp}\dashrightarrow\De_2$, we obtain a
path in $\tableau{S}_0^{\theta}$ of length atmost $n$ that satisfies the
induction hypothesis.

That concludes the induction. 

Getting back to the claim of the Lemma, we have that if
$\neg\co{A}\vp\in\De$, then either $\neg\vp \in \De$ or there exists
an $a'\in A$ such that $\neg \di{a'}\co{A}\vp\in \De$, since $\De$ is
fully expanded. In the former case, $\neg\co{A}\vp$ is realized in
$\De$ and the claims follows. In the latter case, there will be a
prestate $\Ga$ in $\tableau{T}^{\theta}$, such that
$\DtoD{\De}{\Ga}{a'}{\co{A}\vp}$. Note that in this case $\Ga\subseteq
\De$. Due to \RRule{DR} and the fact that $\neg\vp\notin\De$, $\Ga$
cannot contain any of $\neg\co{A}\vp$'s $\beta$-components.

Thus, the statement above gives us that there there is a $\De\rightarrow\Ga\rightarrow\De_1\rightarrow\ldots\rightarrow\De_{n'}$, i.e. there is a path of satisfiable states in $\tableau{S}_0^{\theta}$, that realizes $\neg\co{A}\vp\in\De$. 
\end{proof}

\begin{theorem}
 \label{thr:soundness}
 If $\theta \in \lang$ is satisfiable in a CMAEM, then
 $\tableau{T}^{\theta}$ is open.
\end{theorem}
\begin{proof}
  Using the preceding Lemma \ref{lem:DR_sound} and Lemma
  \ref{lm:E3_sound}, one can show by induction on the number of stages
  in the state elimination process that no satisfiable state can be
  eliminated due to \RRule{E1}--\RRule{E2}.  The claim then follows
  from Lemma~\ref{lem:expansion}.
\end{proof}

%%%%%%%%%%%%%%%%%%%%%%%%%%%%%
\subsection{Completeness}
\label{subsec:completeness}

The \emph{completeness} of a tableau procedure means that if the
tableau for a formula $\theta$ is open, then $\theta$ is satisfiable
in a CMAEM. In view of Theorem~\ref{thr:models_equal_Hintikka}, it
suffices to show that an open tableau for $\theta$ can be turned into
a CMAEHS for $\theta$. In order to prove that, we need to specify
sufficiently strong restrictive conditions \ref{it:newC1} and
\ref{it:newC2} governing the application of the cut rules \ref{it:FE2}
and \ref{it:FE3} respectively on formulas $\distrib{A}\vp$ and
$\commonk{A}\vp$ in the Definition \ref{def:cmaelcd_fully_expanded} of
cut-saturated sets. We now specify these conditions as
follows. 

\begin{enumerate}[label=$C_{\arabic*}$, start=1]
\item \label{it:newC1} \ 
Cut on $\di{A}\vp\in\subf{\psi}$ where $\psi\in\De$, if either of the following holds:
\begin{enumerate}[label= $C_{1\arabic*}$, start=1]
\item  \ $\psi = {\di{B}\delta}$ or $\psi={\neg \di{B}\delta}$, and there is a $\neg\di{E}\varepsilon\in \De$ such that $A\subseteq E$ and $B\subseteq E$. 
\item \ $\psi ={\neg \co{B}\delta}$ and there exists a $\neg\di{E}\varepsilon\in \De$ such that $A\subseteq E$ and $B\cap E\neq \emptyset$. 
\end{enumerate}
  
\item \label{it:newC2} \ 
Cut on $\co{A}\vp\in\subf{\psi}$ where $\psi\in\De$, if either of the following holds:
\begin{enumerate}[label=$C_{2\arabic*}$, start=1]
\item \ $\psi = \di{B}\delta$  or $\psi = \neg \di{B}\delta$, and there exists a $\neg\di{E}\epsilon\in \De$ such that  $B\subseteq E$ and $A\cap E\neq \emptyset$. 
\item \ $\psi = \neg \co{B}\delta$ and there exists a $\neg\di{E}\varepsilon\in \De$ such that $A\cap E\neq \emptyset$ and $B\cap E\neq \emptyset$. 
\end{enumerate}
\end{enumerate}

The intuition: a cut rule only has  to be applied to a formula $\distrib{A}\vp$ or $\commonk{A}\vp$ if: 

\begin{itemize} 
\item[(i)] that formula can occur in the label of a descendant state and, 

\item[(ii)] once it occurs there, it will have an effect spreading back to the current state. 
\end{itemize}

For the former to happen, that formula must occur in a $\distrib{B}$-formula or a $\neg \distrib{B}$-formula or a $\neg\co{B}$-formula. For the latter, the path leading from the current state to that descendant must be labelled with relations propagating the effect of the respective box. 

\begin{example} \label{ex:cut}
  This example illustrates
  the need for applying cut rules and using
  \funny sets instead of simply fully expanded sets.  First, recall
  the requirement of the relations in a (pseudo-)\cmaelcd\ model to be
  equivalence relations, reflected in \ref{it:CHDK} of
  Definition~\ref{def:cmaehs} for Hintikka structures.
  Now, consider the tableau constructed for the formula
  $\theta=\neg\di{\set{a,b}}p\land\neg\di{\set{a,c}}\neg\di{a}p$
  \emph{if we would only use fully expanded sets}:
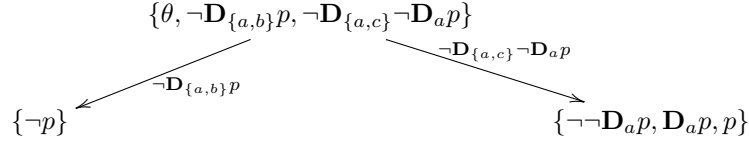
\begin{equation*}
\xymatrix{
&\set{\theta,\neg\di{\set{a,b}}p,\neg\di{\set{a,c}}\neg\di{a}p}\ar^{\neg\di{\set{a,b}}p}[dl]\ar^{\neg\di{\set{a,c}}\neg\di{a}p}[dr]\\
\set{\neg p} &&\set{\neg\neg\di{a}p,\di{a}p,p}
}
\end{equation*}
The corresponding claimed Hintikka structure and (pseudo)-model, that this tableau would produce (see the construction in Lemma \ref{lm:open_tableau_Hintikka}) would then be, respectively:

\medskip
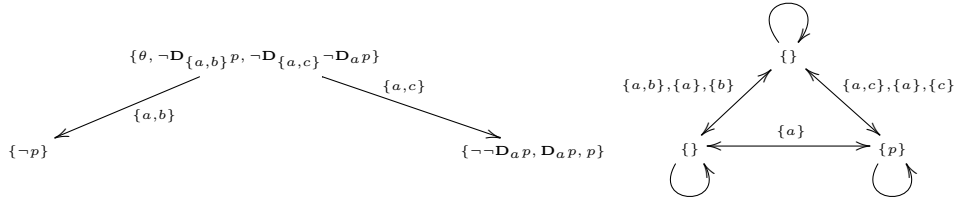
\begin{align*}
\tiny{\xymatrix{
&\set{\theta,\neg\di{\set{a,b}}p,\neg\di{\set{a,c}}\neg\di{a}p}\ar^{\set{a,b}}[dl]\ar^{\set{a,c}}[dr]\\
\set{\neg p} &&\set{\neg\neg\di{a}p,\di{a}p,p}
}}
&
\tiny{\xymatrix{
&\set{}\ar@{<->}[dl]_{\set{a,b},\set{a},\set{b}}\ar@{<->}^{\set{a,c},\set{a},\set{c}}[dr]\ar@(ul,ur)\\
\set{}\ar@{<->}^{\set{a}}[rr]\ar@(dl,dr)&&\set{p}\ar@(dl,dr)
}}
\end{align*}

\bigskip
In the ``Hintikka''-structure to the left, we have that $\di{a}p$ is in the state in the bottom right corner, but not in the state in the top, though the edge connecting them is labelled with $\{a,c\}$. This on the other hand means, that $\theta$ is not satisfied in the ``model'' to the right, because $\di{a}p$ does not hold at  any state, hence $\neg\di{\set{a,c}}\neg\di{a}p$ in not true at any state. In fact, $\theta$ is not satisfiable at all. 

If we would indeed apply the cut-rules then the tableau for $\theta$ would close. The pretableau for $\theta$ would look as follows. 
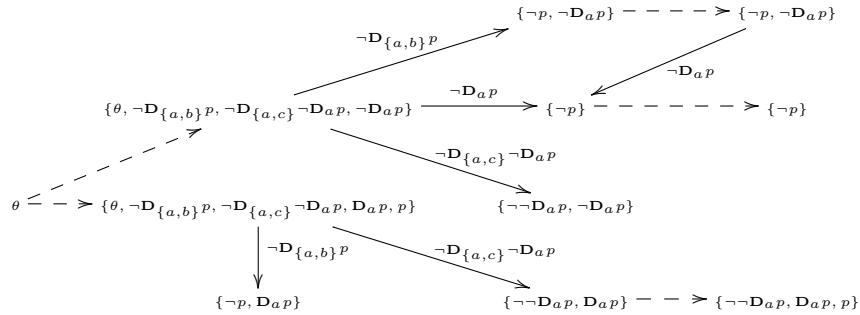
\begin{equation*}
\tiny{
\xymatrix{
&&\set{\neg p,\neg\di{a}p}\ar@{-->}[r]&\set{\neg p,\neg\di{a}p}\ar[dl]^{\neg\di{a}p}\\
&\set{\theta,\neg\di{\set{a,b}}p,\neg\di{\set{a,c}}\neg\di{a}p, \neg\di{a}p}\ar[ru]^{\neg\di{\set{a,b}}p}\ar[dr]^(.7){\neg\di{\set{a,c}}\neg\di{a}p}\ar[r]^(0.7){\neg\di{a}p}
&\set{\neg p}\ar@{-->}[r]&\set{\neg p}\\
\theta\ar@{-->}[ru]\ar@{-->}[r] &\set{\theta,\neg\di{\set{a,b}}p,\neg\di{\set{a,c}}\neg\di{a}p, \di{a}p, p}\ar[d]^{\neg\di{\set{a,b}}p}\ar^(0.7){\neg\di{\set{a,c}}\neg\di{a}p}[rd]&\set{\neg\neg\di{a}p,\neg\di{a}p}\\
&\set{\neg p,\di{a}p} & \set{\neg\neg\di{a}p,\di{a}p}\ar@{-->}[r]&\set{\neg\neg\di{a}p,\di{a}p,p}
}}
\end{equation*}
Notice that some of the prestates (namely $\{\neg p, \di{a}p\}$ and $\{\neg\neg\di{a}p,\neg\di{a}p\}$) do not have any full expansions since these are patently inconsistent. After the initial tableau has been build, this then causes the two states in $\st{\theta}$ to be deleted by \RRule{E1} and the final tableau is 
\begin{equation*}
\tiny{
\xymatrix{
\set{\neg p,\neg\di{a}p}\ar[r]^{\neg\di{a}p} & \set{\neg p} & \set{\neg\neg\di{a}p,\di{a}p,p}
}}
\end{equation*}
which closes. 
\end{example}

\medskip

The following lemma is needed to ensure that the satisfaction of the condition \ref{it:CHDK} from the definition of Hintikka structures for \cmaelcd\ is guaranteed in the final tableau.

\begin{lemma}\label{lemma:CH4}
Suppose $\DtoD{\De}{\De'}{A}{\vp}$ in the final tableau $\tableau{T}^{\theta}$ for some input formula $\theta$ and suppose that $\distrib{B}\psi\in \De'$ where $B\subseteq A$. Then $\distrib{B}\psi\in \De$.
\end{lemma}

\begin{proof}

First, note that if the cut rules \ref{it:FE2} and \ref{it:FE3} are applied unrestrictedly to every  
subformula $\distrib{A}\vp$ or $\commonk{A}\vp$ of a formula in the label of the current state $\De$, the proof of the lemma is immediate. We will show that the claim still holds if the restrictions $C_1$ and $C_2$, specified above, are imposed.

For a formula $\alpha$ we let $\FE_1(\alpha)$ be the set of all formulae that can occur in any one-step \funny expansion of $\alpha$ according to the procedure described in Definition \ref{def:full_expansion}. 
Similarly $\FE_1(\Ga)=\bigcup_{\alpha\in\Ga}\FE_1(\alpha)$ for a set $\Ga$ of formulae, and recursively we let $\FE_n(\Ga)=\FE_1(\FE_{n-1}(\Ga))$.
As is easy to see, this construction converges, and the following  is true:

\begin{itemize}
	\item\label{it:iii} For any formula $\alpha$ and any $n\in \mathbf{N}$, $\FE_n(\alpha)\subseteq \ecl{\alpha}$, i.e.: 
	\begin{align*}
	\FE_n(\alpha)\subseteq \crh{\beta,\neg \beta}{\beta\in\subf{\alpha}}\cup\crh{\distrib{e}\commonk{E}\varepsilon, \neg \distrib{e}\commonk{E}\varepsilon}{\commonk{E}\varepsilon\in\subf{\alpha},e\in E}\end{align*}
	\item For any \funny expansion $\Omega$ of $\Ga$ there is an $n$ such that $\Omega \subseteq \FE_n(\Ga)$. 
	\item\label{it:iv} If $\beta\in\FE_n(\Ga)$, then there is an $\alpha\in \Ga$, such that $\beta\in\FE_n(\alpha)$.
\end{itemize}

Now, let $\Ga$ be the prestate in the pretableau
$\tableau{P}^{\theta}$ that gives rise to the relation between $\De$
and $\De'$, i.e. $\DtoD{\De}{\Ga}{A}{\vp}\dashrightarrow\De'$ in
$\tableau{P}^{\theta}$.  The above gives us that since $\De'$ is a
\funny expansion of $\Ga$ and $\distrib{B}\psi\in\De'$, there is an
$\alpha\in \Ga$ such that $\distrib{B}\psi\in\ecl{\alpha}$.  That is, 
either $\distrib{B}\psi\in \subf{\alpha}$, or $\distrib{B}\psi =
\distrib{d}\commonk{D}\delta$ for a
$\commonk{D}\delta\in\subf{\alpha}$ and a $d\in D$.

Since $\alpha\in\Ga$, due to $\RRule{DR}$, either $\alpha =
\distrib{C}\gamma\in\De$ or $\alpha=\neg \distrib{C}\gamma\in\De$ for
a $C\subseteq A$, or $\alpha = \neg\co{C}\gamma\in\De$ where $C\cap
A\neq\emptyset$, or $\alpha= \neg\vp$. We notice that it is enough to
show that \ref{it:newC1} is applicable to $\di{B}\psi$ at $\De$, since
then either $\distrib{B}\psi\in\De$ (which is what we want) or
$\neg\distrib{B}\psi\in\De$; the latter would, according to
\RRule{DR}, imply that $\neg\distrib{B}\psi\in\Ga\subseteq \De'$,
which would cause $\De'$ to be patently inconsistent, which
contradicts $\De'$ being a \funny set and thus fully expanded
(cf. \ref{it:FE1}).  We split according to cases:

\medskip
\noindent\textbf{Case 1}:   $\alpha = \distrib{C}\gamma\in\De$ or
$\alpha=\neg \distrib{C}\gamma\in\De$ for a $C\subseteq A$:
$\distrib{B}\psi\in \ecl{\distrib{C}\gamma}$ gives that
$\distrib{B}\psi\in\subf{\distrib{C}\gamma}$, or $\distrib{B}\psi =
\distrib{d}\commonk{D}\delta$ for a
$\commonk{D}\delta\in\subf{\distrib{C}\gamma}$, where $d\in D$.
	
In the first case, $\distrib{B}\psi$ is a subformula of a
$\di{C}$-formula in $\De$, and since $C,B\subseteq A$ and
$\neg\di{A}\vp\in\De$, \ref{it:newC1} is applicable to $\di{B}\psi$ at
$\De$.
	
In the second case, $\distrib{B}\psi = \distrib{d}\commonk{D}\delta$
for an $\commonk{D}\delta\in\subf{\distrib{C}\gamma}$ with $d\in
D$. Since $B=\{d\}\subseteq A$, we have $d\in D \inter A$ and hence
\ref{it:newC2} is applicable to $\co{D}\delta$ at $\De$, as we also
have $C\subseteq A$ and $\neg\di{A}\vp\in\De$. This means that either
$\co{D}\delta\in \De$ or $\neg\co{D}\delta \in\De$ according to
\ref{it:FE3}. If $\co{D}\delta\in \De$, then
$\di{B}\psi=\di{d}\co{D}\delta\in\De$ according to \ref{it:FE1}. If
$\neg\co{D}\delta \in\De$, then according to \RRule{DR},
$\neg\co{D}\delta\in\Ga$ since $d\in D\cap A$. However,
$\distrib{B}\psi = \distrib{d}\commonk{D}\delta\in \Ga$, and hence
$\commonk{D}\delta\in \De'$. This gives us a contradiction, as $\De'$
is fully expanded and, thus, not patently inconsistent.

The case where $\alpha=\neg\di{C}\gamma$ is similar.

\medskip
\noindent\textbf{Case 2}:  {$\alpha = \neg \vp$:} $\distrib{B}\psi\in
\ecl{\neg\vp}$. We have two cases to consider:

Either $\distrib{B}\psi\in\subf{\neg\vp}$, in which case
$\distrib{B}\psi\in\subf{\vp}\subseteq\subf{\neg\distrib{A}\vp}$ and
thus \ref{it:newC1} is applicable (since $B,A\subseteq A$).

$\distrib{B}\psi = \distrib{d}\commonk{D}\delta$ for an
$\commonk{D}\delta\in\subf{\neg \vp}$ and $d\in D$ gives that
$\commonk{D}\delta\in\subf{\vp}\subseteq\subf{\neg\distrib{A}\vp}$,
and thus \ref{it:newC2} is applicable to $\co{D}\delta$ at $\De$
since, again, $d\in D\cap A$ and $A\subseteq A$. Then, either
$\commonk{D}\delta\in\De$ or $\neg\commonk{D}\delta\in\De$. As before,
the former implies that $\distrib{B}\psi\in \De$, as desired, while
the latter leads to a contradiction.

\medskip
\noindent\textbf{Case 3}:  $\alpha = \neg\co{C}\gamma$, where $C\cap
A\neq \emptyset$:

$\distrib{B}\psi \in \subf{\neg\co{C}\gamma}$ immediately gives that
\ref{it:newC1} is applicable to $\di{B}\psi$ at $\De$.

If $\distrib{B}\psi = \di{d}\co{D}\delta$, where
$\co{D}\delta\in\subf{\neg\co{C}\gamma}$ and $d\in D$, then
\ref{it:newC2} is applicable to $\co{D}\delta$ at $\De$, as $d\in
D\cap A$ and $\neg\di{A}\vp\in \De$. Thus, either $\co{D}\delta\in
\De$ or $\neg\co{D}\delta$. The former implies that, due to
\ref{it:FE1}, $\di{B}\psi\in\De$, as desired, while the other gives a
contradiction, due to \RRule{DR} and \ref{it:FE1}.
\end{proof}

\begin{lemma}
 \label{lm:open_tableau_Hintikka}
 If $\tableau{T}^{\theta}$ is open, then there exists a CMAEHS for
 $\theta$.
\end{lemma}

\begin{proof}

The needed Hintikka structure \hintikka{H} for the formula $\theta$ is built out of the final tableau $\tableau{T}^{\theta}$ by renaming the relations between the states, such that they correspond to a subset of $\agents$. This is done by labeling the edges from $\De$ to $\De'$ with the set $A$ for which $\DtoD{\De}{\De'}{A}{\varphi}$ in $\tableau{T}^{\theta}$.

Now, let $\agents$ be the set of agents occurring in $\theta$, and let $S = 
\tableau{S}^{\theta}$. 
For any $A\in\powerne{\agents}$, let $\rel{R}^D_A = \crh{(\De,\De')\in
  S\times S}{\DtoD{\De}{\De'}{A}{\varphi}\text{ for some }\varphi}$,
and let $\rel{R}^C_A$ be the reflexive, transitive closure of
$\bigunion_{B\subseteq A}\, \rel{R}^D_{B}$. Let $L(\De)$ be the
labelling of the state in $\tableau{T}$, i.e. the sets of formulae
that has been associated with $\De$.

Finally, let $\hintikka{H}_{\theta} = (\agents, \tableau{S}^{\theta},
\set{\rel{R}^D_A}_{A \in \powerne{\agents}}, \set{\rel{R}^C_A}_{A \in
  \powerne{\agents}}, \ap, L )$.

We will now show that $\hintikka{H}_{\theta}$ is a Hintikka
structure. To that end, we have to prove $(\agents, S,
\set{\rel{R}^D_A}_{A \in \powerne{\agents}}, \set{\rel{R}^C_A}_{A \in
  \powerne{\agents}})$ is a CMAES, and that
conditions~\ref{it:CHDS}-\ref{it:CHnegCK} of
Definition~\ref{def:cmaehs} hold for $\hintikka{H}$.  The former is
clear from the construction of $\hintikka{H}$.

\ref{it:CHDS} holds since all states in the final tableau are fully
expanded.  

\ref{it:CHnegDK} is satisfied since, otherwise, the state
would have been deleted from the tableau due to \RRule{E1}.  

Likewise, \ref{it:CHnegCK} is satisfied since, otherwise, the state would have
been removed due to \RRule{E2}. 

It remains to show that \ref{it:CHDK} holds. Let $(\De, \De')\in
\rel{R}^D_A$ (i.e. $\DtoD{\De}{\De'}{A}{\vp}$ in
$\tableau{T}^{\theta}$), and $B\subseteq A$. We need to show that
$\distrib{B}\psi\in \De \Leftrightarrow \distrib{B}\psi\in
\De'$.  If $\distrib{B}\psi\in \De$, then due to the propagation
rule \RRule{DR}, $\distrib{B}\psi\in \Ga$, where $\Ga$ is the prestate
in the final pretableau, such that
$\DtoD{\De}{\Ga}{A}{\vp}\dashrightarrow\De'$. Thus $\distrib{B}\psi$
is also in $\De'$ since $\Ga$ is included in all \funny expansions of
$\Ga$.  The other direction follows from Lemma \ref{lemma:CH4}
\end{proof}

\begin{theorem}[Completeness]
 \label{thr:completeness}
 Let $\theta \in \lang$ and let $\tableau{T}^{\theta}$ be open.
 Then, $\theta$ is satisfiable in a CMAEM.
\end{theorem}
\begin{proof}
 Immediate from Lemma~\ref{lm:open_tableau_Hintikka} and
 Theorem~\ref{thr:models_equal_Hintikka}. 
\end{proof}

%%%%%%%%%%%%%%%%%%%%%%%%%%%%%%%%%%%%%%%%%%%%%%
%%%%%%%%%%%%%%%%%%%%%%%%%%%%%%%%%%%%%%%%%%%%%%
\section{Complexity, efficiency, and possible optimizations of the tableau procedure}
\label{sec:complexity}

\subsection{Complexity}

The termination of the tableau procedure described above is a fairly straightforward consequence of the finiteness of the set of all possible labels of states and prestates and their re-use in the construction phase. 
In this subsection, we estimate the worst-case running time of all phases of the procedure.

We denote by $|\theta|$ the length of a formula $\theta$
and by $\card{\ecl{\theta}}$ the number of formulae in $\ecl{\theta}$.
Let $|\theta | = n$ and the number of agents occurring in $\theta$ be $k$. 

By Lemma~\ref{lem:size_of_closure}, $\card{\ecl{\theta}} \leq c k n$ for some (small) constant $c$. Then, the number of prestates and states in the tableau for $\theta$ is 
$\bigo{2^{ckn}}$. Comparing two states or prestates takes  \bigo{ckn} steps (assuming a fixed order of the formulae in $\ecl{\theta}$, and each state being represented as a 0/1 string of length $ckn$), hence checking whether a prestate or a state has already been created, takes $\bigo{ckn 2^{ckn}}$.  Therefore, the construction phase takes time $\bigo{ckn2^{2ckn}}$.

The prestate elimination phase takes time $\bigo{2^{ckn}}$.  Checking realization of an eventuality in a state takes $\bigo{2^{ckn}}$ steps and the number of eventualities is bounded by $n$, hence  the elimination of a `bad' state takes at most $\bigo{n2^{ckn}}$ steps. Hence, the elimination state takes 
$\bigo{n2^{2ckn}}$ steps.  

We conclude that the whole tableau procedure 
terminates in $\bigo{ckn2^{2ckn}}$ steps, hence it is in EXPTIME, which is in compliance with the known EXPTIME(-complete) lower bound (see \cite{Fagin95knowledge}, \cite{FHV92}).

\subsection{Efficiency}

Some features of the ``diamond-propagating'' procedure described above
make it sometimes practically sub-efficient.

Firstly, the application of the cut rules \ref{it:FE2} and \ref{it:FE3}
can produce many \funny sets, even after imposing the restrictive
conditions $C_1$ and $C_2$. Potentially, it can create a number of
states that is exponential in the number of subformulae of the form
$\distrib{A} \psi$ or $\commonk{A}\psi$ occurring in the formulas of
the input set $\Ga$.

Secondly, when applying the rule \RRule{DR} to a state $\De$ with respect to
some $\neg \distrib{A} \vp$, we propagate to the newly created
prestate all the diamond-formulae of the form $\neg \distrib{B} \psi$,
where $B \subseteq A$, except $\neg \distrib{A} \vp$ itself. Likewise,
all formulae $\neg\commonk{A}\psi$ where $A$ and $B$ are not disjoint,
get propagated. Thus, the presence of a ``diamond'' in a prestate
$\Ga$ is then passed on to all states in $\st{\Ga}$, resulting in the need to apply the rule \RRule{DR} to every state in $\st{\Ga}$ with
respect to this diamond; this, again, implies the creation of a large
number of states (even though, as we have shown, the maximal number of
states is still no more than exponential in the size of the input
formula). However, we re-iterate that this `diamond-propagation' is necessary
for the procedure developed here, because if a diamond-formula is not
propagated forward, then its negation, which is a box-formula, may be
added to a successor state and thus clash with that diamond-formula in
the current state.

On the other hand, the restrictive conditions $C_1$ and $C_2$  for the application of cut-saturation in the production of CS-expansions can have a very significant effect on the size of the tableau, as illustrated by the next example. 

\medskip

\begin{example}
Suppose we want to build a tableau for
the formula $\theta = \commonk{\set{a,b}} \distrib{a} p \imp \neg
\commonk{\set{b,c}} \distrib{b} p \equiv \neg(\commonk{\set{a,b}} \distrib{a} p \land
\commonk{\set{b,c}} \distrib{b} p)  $ and suppose that $\agents = \set{a,
  b,c}$.  We start off with creating a single prestate \set{\theta}. Using only the unrestricted conditions $C_1$ and $C_2$ to cut, 
applying the rule \RRule{SR} to this prestate produces an overwhelming number of 35 states:
\begin{enumerate}
\small{
\item \set{\theta, {\neg \commonk{\set{a,b}} \distrib{a} p}, 
		\neg\distrib{a} \commonk{\set{a,b}}\distrib{a} p, \distrib{a} p, p, \underline{\commonk{\set{b,c}}\distrib{b}p}}; 
\item \set{\theta, {\neg \commonk{\set{a,b}} \distrib{a} p}, \neg
    \distrib{a}\commonk{\set{a,b}} \distrib{a} p, \distrib{a} p, p, {\neg \commonk{\set{b,c}}\distrib{b}p}, 
    \neg\distrib{b}\commonk{\set{b,c}}\distrib{b}p, \distrib{b}p};
\item \set{\theta, {\neg \commonk{\set{a,b}} \distrib{a} p}, \neg
    \distrib{a} \commonk{\set{a,b}}\distrib{a} p, \distrib{a} p, p, {\neg \commonk{\set{b,c}}\distrib{b}p}, 	
    \neg\distrib{b}\commonk{\set{b,c}}\distrib{b}p, \neg\distrib{b}p};
\item \set{\theta, {\neg \commonk{\set{a,b}} \distrib{a} p}, \neg
    \distrib{a}\commonk{\set{a,b}} \distrib{a} p, \distrib{a} p, p, {\neg \commonk{\set{b,c}}\distrib{b}p}, 
    \neg\distrib{c}\commonk{\set{b,c}}\distrib{b}p, \distrib{b}p};
\item \set{\theta, {\neg \commonk{\set{a,b}} \distrib{a} p}, \neg
    \distrib{a} \commonk{\set{a,b}}\distrib{a} p, \distrib{a} p, p, {\neg \commonk{\set{b,c}}\distrib{b}p}, 	
    \neg\distrib{c}\commonk{\set{b,c}}\distrib{b}p, \neg\distrib{b}p};
\item \set{\theta, {\neg \commonk{\set{a,b}} \distrib{a} p}, \neg
    \distrib{a} \commonk{\set{a,b}}\distrib{a} p, \distrib{a} p, p, {\neg \commonk{\set{b,c}}\distrib{b}p}, 	
    \neg\distrib{b}p};

\item \set{\theta, {\neg \commonk{\set{a,b}} \distrib{a} p}, \neg
    \distrib{a}\commonk{\set{a,b}} \distrib{a} p, \neg \distrib{a} p, \underline{\commonk{\set{b,c}}\distrib{b}p}, p};
\item \set{\theta, {\neg \commonk{\set{a,b}} \distrib{a} p}, \neg
    \distrib{a}\commonk{\set{a,b}} \distrib{a} p, \neg \distrib{a} p, {\neg \commonk{\set{b,c}}\distrib{b}p}, 
    \neg\distrib{b}\commonk{\set{b,c}}\distrib{b}p, \distrib{b}p, p};
\item \set{\theta, {\neg \commonk{\set{a,b}} \distrib{a} p}, \neg
    \distrib{a}\commonk{\set{a,b}} \distrib{a} p, \neg \distrib{a} p, {\neg \commonk{\set{b,c}}\distrib{b}p}, 
    \neg\distrib{b}\commonk{\set{b,c}}\distrib{b}p, \neg\distrib{b}p};
\item \set{\theta, {\neg \commonk{\set{a,b}} \distrib{a} p}, \neg
    \distrib{a}\commonk{\set{a,b}} \distrib{a} p, \neg \distrib{a} p, {\neg \commonk{\set{b,c}}\distrib{b}p}, 
    \neg\distrib{c}\commonk{\set{b,c}}\distrib{b}p, \distrib{b}p, p};
\item \set{\theta, {\neg \commonk{\set{a,b}} \distrib{a} p}, \neg
    \distrib{a}\commonk{\set{a,b}} \distrib{a} p, \neg \distrib{a} p, {\neg \commonk{\set{b,c}}\distrib{b}p}, 
    \neg\distrib{c}\commonk{\set{b,c}}\distrib{b}p, \neg\distrib{b}p};
\item \set{\theta, {\neg \commonk{\set{a,b}} \distrib{a} p}, \neg
    \distrib{a}\commonk{\set{a,b}} \distrib{a} p, \neg \distrib{a} p, {\neg \commonk{\set{b,c}}\distrib{b}p}, 
 \neg\distrib{b}p};
    
\item \set{\theta, {\neg \commonk{\set{a,b}} \distrib{a} p}, 
		\neg\distrib{b} \commonk{\set{a,b}}\distrib{a} p, \distrib{a} p, p, \underline{\commonk{\set{b,c}}\distrib{b}p}}; 
\item \set{\theta, {\neg \commonk{\set{a,b}} \distrib{a} p}, \neg
    \distrib{b}\commonk{\set{a,b}} \distrib{a} p, \distrib{a} p, p, {\neg \commonk{\set{b,c}}\distrib{b}p}, 
    \neg\distrib{b}\commonk{\set{b,c}}\distrib{b}p, \distrib{b}p};
\item \set{\theta, {\neg \commonk{\set{a,b}} \distrib{a} p}, \neg
    \distrib{b} \commonk{\set{a,b}}\distrib{a} p, \distrib{a} p, p, {\neg \commonk{\set{b,c}}\distrib{b}p}, 	
    \neg\distrib{b}\commonk{\set{b,c}}\distrib{b}p, \neg\distrib{b}p};
\item \set{\theta, {\neg \commonk{\set{a,b}} \distrib{a} p}, \neg
    \distrib{b}\commonk{\set{a,b}} \distrib{a} p, \distrib{a} p, p, {\neg \commonk{\set{b,c}}\distrib{b}p}, 
    \neg\distrib{c}\commonk{\set{b,c}}\distrib{b}p, \distrib{b}p};
\item \set{\theta, {\neg \commonk{\set{a,b}} \distrib{a} p}, \neg
    \distrib{b} \commonk{\set{a,b}}\distrib{a} p, \distrib{a} p, p, {\neg \commonk{\set{b,c}}\distrib{b}p}, 	
    \neg\distrib{c}\commonk{\set{b,c}}\distrib{b}p, \neg\distrib{b}p};
\item \set{\theta, {\neg \commonk{\set{a,b}} \distrib{a} p}, \neg
    \distrib{b} \commonk{\set{a,b}}\distrib{a} p, \distrib{a} p, p, {\neg \commonk{\set{b,c}}\distrib{b}p}, 	
    \neg\distrib{b}p};

\item \set{\theta, {\neg \commonk{\set{a,b}} \distrib{a} p}, \neg
    \distrib{b}\commonk{\set{a,b}} \distrib{a} p, \neg \distrib{a} p, \underline{\commonk{\set{b,c}}\distrib{b}p}, p};
\item \set{\theta, {\neg \commonk{\set{a,b}} \distrib{a} p}, \neg
    \distrib{b}\commonk{\set{a,b}} \distrib{a} p, \neg \distrib{a} p, {\neg \commonk{\set{b,c}}\distrib{b}p}, 
    \neg\distrib{b}\commonk{\set{b,c}}\distrib{b}p, \distrib{b}p, p};
\item \set{\theta, {\neg \commonk{\set{a,b}} \distrib{a} p}, \neg
    \distrib{b}\commonk{\set{a,b}} \distrib{a} p, \neg \distrib{a} p, {\neg \commonk{\set{b,c}}\distrib{b}p}, 
    \neg\distrib{b}\commonk{\set{b,c}}\distrib{b}p, \neg\distrib{b}p};
\item \set{\theta, {\neg \commonk{\set{a,b}} \distrib{a} p}, \neg
    \distrib{b}\commonk{\set{a,b}} \distrib{a} p, \neg \distrib{a} p, {\neg \commonk{\set{b,c}}\distrib{b}p}, 
    \neg\distrib{c}\commonk{\set{b,c}}\distrib{b}p, \distrib{b}p, p};
\item \set{\theta, {\neg \commonk{\set{a,b}} \distrib{a} p}, \neg
    \distrib{b}\commonk{\set{a,b}} \distrib{a} p, \neg \distrib{a} p, {\neg \commonk{\set{b,c}}\distrib{b}p}, 
    \neg\distrib{c}\commonk{\set{b,c}}\distrib{b}p, \neg\distrib{b}p};
\item \set{\theta, {\neg \commonk{\set{a,b}} \distrib{a} p}, \neg
    \distrib{b}\commonk{\set{a,b}} \distrib{a} p, \neg \distrib{a} p, {\neg \commonk{\set{b,c}}\distrib{b}p}, 
    \neg\distrib{b}p};

\item \set{\theta, {\neg \commonk{\set{a,b}} \distrib{a} p}, \neg \distrib{a} p, \underline{\commonk{\set{b,c}}\distrib{b}p}, p};
\item \set{\theta, {\neg \commonk{\set{a,b}} \distrib{a} p}, \neg \distrib{a} p, {\neg \commonk{\set{b,c}}\distrib{b}p}, 
    \neg\distrib{b}\commonk{\set{b,c}}\distrib{b}p, \distrib{b}p, p};
\item \set{\theta, {\neg \commonk{\set{a,b}} \distrib{a} p}, \neg \distrib{a} p, {\neg \commonk{\set{b,c}}\distrib{b}p}, 
    \neg\distrib{b}\commonk{\set{b,c}}\distrib{b}p, \neg\distrib{b}p};
\item \set{\theta, {\neg \commonk{\set{a,b}} \distrib{a} p}, \neg \distrib{a} p, {\neg \commonk{\set{b,c}}\distrib{b}p}, 
    \neg\distrib{c}\commonk{\set{b,c}}\distrib{b}p, \distrib{b}p, p};
\item \set{\theta, {\neg \commonk{\set{a,b}} \distrib{a} p}, \neg \distrib{a} p, {\neg \commonk{\set{b,c}}\distrib{b}p}, 
    \neg\distrib{c}\commonk{\set{b,c}}\distrib{b}p, \neg\distrib{b}p};
\item \set{\theta, {\neg \commonk{\set{a,b}} \distrib{a} p}, \neg \distrib{a} p, {\neg \commonk{\set{b,c}}\distrib{b}p}, 
    \neg\distrib{b}p};

\item \set{\theta, \underline{\commonk{\set{a,b}} \distrib{a} p}, p, {\neg \commonk{\set{b,c}}\distrib{b}p}, 
    \neg\distrib{b}\commonk{\set{b,c}}\distrib{b}p, \distrib{b}p};
\item \set{\theta, \underline{\commonk{\set{a,b}} \distrib{a} p}, p, {\neg \commonk{\set{b,c}}\distrib{b}p}, 
    \neg\distrib{b}\commonk{\set{b,c}}\distrib{b}p, \neg\distrib{b}p};
\item \set{\theta, \underline{\commonk{\set{a,b}} \distrib{a} p}, p, {\neg \commonk{\set{b,c}}\distrib{b}p}, 
    \neg\distrib{c}\commonk{\set{b,c}}\distrib{b}p, \distrib{b}p};
\item \set{\theta, \underline{\commonk{\set{a,b}} \distrib{a} p}, p, {\neg \commonk{\set{b,c}}\distrib{b}p}, 
    \neg\distrib{c}\commonk{\set{b,c}}\distrib{b}p, \neg\distrib{b}p};
\item \set{\theta, \underline{\commonk{\set{a,b}} \distrib{a} p}, p, {\neg \commonk{\set{b,c}}\distrib{b}p}, 
    \neg\distrib{b}p};   }
\end{enumerate}
If we instead use the restricted \ref{it:newC1} and \ref{it:newC2}, we will produce 8 states:
\begin{enumerate}
\small{
\item \set{\theta, {\neg \commonk{\set{a,b}} \distrib{a} p}, \neg
    \distrib{a}\commonk{\set{a,b}} \distrib{a} p, \distrib{a} p, p}
\item \set{\theta, {\neg \commonk{\set{a,b}} \distrib{a} p}, \neg
    \distrib{a}\commonk{\set{a,b}} \distrib{a} p, \neg \distrib{a} p}
\item \set{\theta, {\neg \commonk{\set{a,b}} \distrib{a} p}, \neg
    \distrib{b}\commonk{\set{a,b}} \distrib{a} p}
\item \set{\theta, {\neg \commonk{\set{a,b}} \distrib{a} p}, \neg
    \distrib{a} p}
\item \set{\theta, {\neg \commonk{\set{b,c}} \distrib{b} p}, \neg
    \distrib{b}\commonk{\set{b,c}} \distrib{b} p, \distrib{b} p, p}
\item \set{\theta, {\neg \commonk{\set{b,c}} \distrib{b} p}, \neg
    \distrib{b}\commonk{\set{b,c}} \distrib{b} p, \neg \distrib{b} p}
\item \set{\theta, {\neg \commonk{\set{b,c}} \distrib{b} p}, \neg
    \distrib{c}\commonk{\set{b,c}} \distrib{b} p}
\item \set{\theta, {\neg \commonk{\set{b,c}} \distrib{b} p}, \neg
    \distrib{b} p}
}
\end{enumerate}
Figure \ref{fig:example2} shows the pretableau for one part of $\theta$, i.e. $\neg\co{\set{a,b}}\di{a}p$. The tableau for the other part of $\theta$ will be similar and disjoint from this tableau.
\begin{figure}[!ht]
\begin{minipage}[c]{0.58\linewidth}
\begin{equation*}
\tiny{
\xymatrix{
&&\Ga_0\ar@{-->}[dl]\ar@{-->}[d]\ar@{-->}[dr]\ar@{-->}[drr]\\
 &\De_2\ar[dl]_{\chi_0}\ar[d]^{\chi_a}& \De_1 \ar[d]_{\chi_0} & \De_3\ar^(0.4){\chi_b}[dddr] & \De_4\ar^{\chi_a}[d] \\ 
\Ga_3\ar@{-->}[ddr]& \Ga_2\ar@{-->}[dddr] & \Ga_1 \ar@{-->}[ddl]\ar@{-->}[d]\ar@{-->}[dr] & & \Ga_4\ar@{-->}[dddl]\ar@{-->}[d] \\
 & &\De_5 \ar@/_0.6pc/_(0.4){\chi_0}[u] & \De_6\ar^(0.15){\chi_b}[dr]& \De_8\ar^{\chi_b}[d]\\
&\De_7\ar_{\chi_a}[uu]\ar@/_0.5pc/_(.75){\chi_0}[uul]&&&\Ga_5 \ar@{-->}[dllll]\ar@{-->}[dll]\ar@{-->}[dl]\ar@{-->}[d]\\
\De_{11}\ar_(0.4){\chi_a}[uuur]\ar^{\chi_0}[uuu]&&\De_{9} \ar@/^0.8pc/_{\chi_0}[uuu] & \De_{12}\ar@/^0.5pc/^(0.35){\chi_a}[uuur]&\De_{10}\ar@/_0.75pc/_{\chi_b}[u]
}
}
\end{equation*}
\end{minipage}%
\hfill
\begin{minipage}[c]{0.40\linewidth}
\tiny{
\begin{align*}
\chi_0 & = \neg\di{a}p, \chi_a = \neg\di{a}\co{\set{a,b}}\di{a}p, \chi_b = \neg\di{b}\co{\set{a,b}}\di{a}p\\
\\
\Ga_0 &=\set{\neg\co{\set{a,b}}\di{a}p}\\
\De_1 &=\set{{\neg\co{\set{a,b}}\di{a}p},\neg\di{a}p }\\
\De_2 &=\set{{\neg\co{\set{a,b}}\di{a}p},\neg\di{a}\co{\set{a,b}}\di{a}p, \neg\di{a}p}\\
\De_3 &=\set{{\neg\co{\set{a,b}}\di{a}p},\neg\di{b}\co{\set{a,b}}\di{a}p}\\
\De_4 &=\set{{\neg\co{\set{a,b}}\di{a}p},\neg\di{a}\co{\set{a,b}}\di{a}p, \di{a}p, p}\\
\Ga_1 &=\set{\neg p, \di{a}\neg\co{\set{a,b}}\di{a}p}\\
\Ga_2 &=\set{\neg\co{\set{a,b}}\di{a}p,\neg\di{a}p }\\
\Ga_3 &=\set{\neg p, \di{a}\neg\co{\set{a,b}}\di{a}p, \neg\di{a}\co{\set{a,b}}\di{a}p}\\
\Ga_4 &=\set{\neg\co{\set{a,b}}\di{a}p, \di{a}p }\\
\De_5 &= \set{\neg p, {\neg\co{\set{a,b}}\di{a}p}, \neg\di{a}p}\\
\De_6 &= \set{\neg p, {\neg\co{\set{a,b}}\di{a}p}, \neg\di{b}\co{\set{a,b}}\di{a}p}\\
\De_7 &= \set{\neg p, {\neg\co{\set{a,b}}\di{a}p}, \neg\di{a}\co{\set{a,b}}\di{a}p, \neg\di{a}p}\\
\De_8 &=\set{{\neg\co{\set{a,b}}\di{a}p}, \di{a}p, p, \neg \di{b}\co{\set{a,b}}\di{a}p }\\
\Ga_5 &=\set{\neg\co{\set{ab}}\di{a}p}\\
\De_{9} &= \set{{\neg\co{\set{a,b}}\di{a}p}, \neg\di{a}p}\\
\De_{10} &= \set{{\neg\co{\set{a,b}}\di{a}p}, \neg\di{b}\co{\set{a,b}}\di{a}p}\\
\De_{11} &= \set{{\neg\co{\set{a,b}}\di{a}p}, \neg\di{a}\co{\set{a,b}}\di{a}p,\neg\di{a}p}\\
\De_{12} &= \set{{\neg\co{\set{a,b}}\di{a}p}, \neg\di{a}\co{\set{a,b}}\di{a}p,\di{a}p, p}\\
\end{align*}}
\end{minipage}
\caption{Pretableau for $\neg\co{\set{a,b}}\di{a}p$}
\label{fig:example2}
\end{figure}

As seen here, the backtracking procedure is rather inefficient when
applied to formulae of the type of $\commonk{\set{a,b}} \distrib{a} p \imp \neg \commonk{\set{b,c}} \distrib{b} p$. 
\hfill$\Box$\\
\end{example}

Both causes of potential inefficiencies discussed above, viz. the forward diamond-propagation and the (restricted) analytic cut rules on box-formulae, are needed to ensure that every ``successful'' tableau
can be turned into a Hintikka structure.  More precisely, they together ensure that the 
right-to-left implication in the statement of property (CH3) of
Hintikka structures (recall Definition~\ref{def:cmaehs}) holds.

A possible way of eliminating these causes for inefficiencies is to change the strategy in the tableau-building, by  implementing a mechanism for \emph{backward propagation of boxes}: if $\distrib{A} \vp$ occurs in a state $\De$, then ensure that this box is propagated backwards to all predecessor states where it must occur. The main disadvantage of this approach is that it requires an elaborated mechanism of repeated updating the hitherto constructed part of the tableau. We leave the realization of this idea for future work.

\subsection{Improvements}
\label{subsec:improvements}

As stated earlier, the main emphasize of our tableau construction is the ease of presentation, comprehension and implementation, rather than technical sophistication and optimality of the procedure. While being worst-case time optimal, it is amenable to various improvements and further optimizations, some of which we will mention briefly here.  

To begin with, for methodological reasons, our procedure is divided into three phases, where the different components of the tableau-building procedure are dealt with separately. That separation of the procedure into phases makes it less optimal compared to the approach whereby the three phases are carried out simultaneously and the prestate and state elimination is done `on-the-fly'. 

Also, as briefly mentioned in Section \ref{subsec:overview_of_procedure}, it is possible to make the procedure cut-free by using a mechanism for `backwards propagation' of $\di{}$-formulas, which, when well designed can lead to more optimal performance in some cases. This approach is taken e.g., in \cite{GoreWidmann10}, where the authors construct a cut-free tableau-based algorithm for the logic PDL with converse, while the algorithm presented in \cite{Nguyen11} builds on this work by constructing a cut-free tableaux-based algorithm for the description logic SHI, which contains inverse roles. Both methods account for the case where a (number of) formula(s) turns up in a node, which will be required to be in the already created predecessor node of the node in question. The former algorithm deals with eventualities, too.  Adopting this approach to our procedure while optimizing it for the logic \CMAELCD\ would result in a procedure sketched below. 

\subsubsection{State elimination `on-the-fly'}
Here we make use of the concept of `potential rescuers' used in \cite{GoreWidmann10} and \cite{Nguyen11}, though in a slightly different way,  adjusted to our needs. 
We likewise take on board the techniques of updating and propagating statuses of nodes in the tableaux. 

Firstly, we maintain a status for (pre)states, which can either be $\ttt{unexplored}$, 
$\ttt{open}$ or $\ttt{closed}$. The status of a (pre)state is initially set to $\ttt{unexplored}$ when the (pre)state is created, and then updated during the procedure. When a prestate is expanded or a state expanded for all diamond-formulas in it, its status changes to $\ttt{open}$. Later on the status of a state $\De$ can then change to $\ttt{closed}$ in the following cases: 
\begin{itemize}
\item there is an epistemic prestate $\Ga$ such that  $\De\overset{\delta}\rightarrow \Ga$ for a formula $\delta$ and the status of $\Ga$ is $\ttt{closed}$. 
\item $\De$ contains an eventuality $\neg\co{A}\vp$ that it neither realized in the current tableau under construction nor has a ``potential rescuer''. A potential rescuer is a (pre)state, which 
is $A$-reachable from $\De$, contains $\neg\co{A}\vp$, and has not been expanded yet, i.e. it has status $\ttt{unexplored}$. Here we use a modified definition of $A$-reachability, where $\dashrightarrow$-arrows are allowed too.
\end{itemize} 
The status of a prestate $\Ga$ is set to $\ttt{closed}$ if: 
\begin{itemize}
\item all states in $\st{\Ga}$ are closed, including the case where $\st{\Ga}=\emptyset$, or 
\item $\Ga$ contains an eventuality that it neither realized nor has a potential rescuer.
\end{itemize}
Additionally, we make sure, that unsatisfiable (pre)states are removed on-the-fly and that the procedure stops and the tableau closes as soon as unsatisfiability of the input prestate is detected during the procedure, i.e.:

\begin{itemize}
\item We close a prestate when it is expanded and does not have any cut-saturated expansions. 
\item When a (pre)state $\Sigma$ closes, we propagate updates of statuses to the relevant (pre)states, whose status depend on the status of $\Sigma$. These are (pre)states that have outgoing arrows pointing to $\Sigma$.
\item We keep an eye on the initial prestate, labelled with the input formula whose satisfiability we are checking. When/if this prestate closes, we stop the whole procedure and return ``unsat''.
\end{itemize}
Finally, we also want to avoid the unnecessary checking of unrealized eventualities, since this step is one of the more expensive checks. Thus, when updating the status of a (pre)state we only check containment of unrealized eventualities, when this is really necessary. E.g., we do not check that if a potential rescuer is known to be reachable. 
This of course requires some bookkeeping.

\subsubsection{Making the procedure cut-free}
The procedure above takes care of doing the satisfiability checking `on the fly', however it is not cut-free. Though, the procedure can be made cut-free by incorporating the following:

Firstly, we use full expansions instead of cut-saturated expansions. 
Secondly, we now need to account for a further reason why a state $\De$ can close, namely that $\De$ contains a diamond formula $\neg \di{A}\vp$, 
such that $\DtoD{\De}{\Ga}{A}{\vp}$ and all states in $\st{\Ga}$ are \de{incompatible} with $\De$ with respect to $\neg\di{A}\vp$. Here, $\De'\in\st{\Ga}$ is incompatible with $\De$ if $\crh{\di{A'}\psi\in\De'}{A'\subseteq A}\not\subseteq \De$, i.e. condition %CH3
\ref{it:CHDK} will not be fulfilled in the resulting Hintikka structure. 
This, however, does not neccessarily mean, that the state $\De$ needs to close. After all, since we are not proactively looking ahead for box- formulas which could possibly occur in a future descendent state of $\De$ and include these in $\De$ (as is done when using cut-saturated expansions), it is possible that $\De$ could become satisfiable if the box-formulas in question were added to $\De$. 

Therefore, when it happens that $\DtoD{\De}{\Ga}{A}{\vp}$ and none of the states in $\st{\Ga}$ are compatible with $\De$ with respect to $\neg\di{A}\vp$, we construct so-called `alternatives' for the state $\De$.
These are states labelled with the fully expanded sets 
$\De \cup S'$ for each $S'\in \bigcup_{\De'\in\; \st{\Ga}} \mathcal{FE}(\crh{\di{A'}\psi\in\De'}{A'\subseteq A \text{ and } \di{A'}\psi\notin\De})$. 
Then $\dashrightarrow$-arrows pointing to these alternatives are added from each prestates pointing to $\De$, and finally we close the original state $\De$ (and propagate the change of status that hereby occurs, as described previously).  

In this procedure, we need to keep track of when such incompatibilities occur, which requires some further bookkeeping.

%%%%%%%%%%%%%%%%%%%%%%%%%%%%%%%%%%
\section{Concluding remarks}
\label{sec:concluding}

We have developed a sound and complete tableau-based decision
procedure for the full coalitional multiagent epistemic logic
\CMAELCD. The incremental tableau style adopted here is 
intuitive, practically more efficient, and more flexible than the
maximal tableau style, developed e.g., for the fragment \MAELC\ of \CMAELCD\
in~\cite{HM92}, and therefore it is more suitable both for manual
and automated execution. In fact, an earlier, less optimal version, of
this procedure has been implemented and reported in
\cite{VestergaardThesis}. On the other hand, as discussed in the previous section, various further optimizations of the procedure are possible and desirable, and some such optimizations have been developed for logics related to \CMAELCD, see Section \ref{subsec:related}. Furthermore, our tableau procedure is also amenable to various extensions, subject to current and future work:
\begin{itemize}
\item to temporal epistemic logics of linear and branching time,
  preliminary reports on which have appeared respectively in
  \cite{GorSh09a} and \cite{DBLP:conf/mallow/GorankoS09}.
\item with the strategic abilities operators of the Alternating-time
  temporal logic $\ATL$, a tableau-based decision procedure for which
  were developed in \cite{GorSh08}. Merging tableaux for these two logical systems will produce, inter alia, a feasible decision procedure for the Alternating-time temporal epistemic logic $\ATEL$ \cite{vdHW04}.
  \item a cut-free, `on the fly' version, as described in Section \ref{subsec:improvements}.

\end{itemize}

\section*{Acknowledgments}
We gratefully acknowledge the financial support from: the FIRST Research School and Roskilde University, funding the PhD study of the first author, part of which is her contribution to the present research; 
the National Research Foundation of South Africa for the second author through several research grants, and the Claude Harris Leon Foundation, funding the third author's post-doctoral fellowship at the University of the Witwatersrand, during which the initial work of this research was done. 
We also thank the anonymous referees for their valuable comments and constructive criticism.

\bibliographystyle{plain}
\bibliography{ctl-mas}
%\end{paper}
\end{document}